%% file: main.tex
\documentclass[11pt]{article}
\usepackage{amsmath,amssymb,amsthm,amsfonts,verbatim}
\usepackage[colorlinks=true,
            linkcolor=red,
            urlcolor=violet,
            citecolor=blue]{hyperref}
\usepackage{fullpage,geometry,xcolor}
\usepackage{dsfont}

\newcommand{\floor}[1]{\left\lfloor #1 \right\rfloor}

\newcommand{\var}{\textrm{Var}}

\newcommand{\eat}[1]{}
\newcommand{\calG}{\ensuremath{\mathcal{G}}}

\newcommand{\E}{\mathop{\mathbb{E}\/}}

\newcommand{\poly}{\ensuremath{\mathrm{poly}}}

\newcommand{\prg}{{\ensuremath \mathsf{PRG}}}
\newcommand{\bpp}{{\ensuremath \mathsf{BPP}}}
\newcommand{\rl}{{\ensuremath \mathsf{RL}}}
\newcommand{\logspace}{{\ensuremath \mathsf{L}}}
\newcommand{\ptime}{{\ensuremath \mathsf{P}}}
\newcommand{\tc}{{\ensuremath \mathsf{TC}_0}}

\newcommand{\eps}{\ensuremath{\varepsilon}}
\renewcommand{\epsilon}{\ensuremath{\varepsilon}}
\newcommand{\zo}{\{0,1\}}
\newcommand{\pmo}{\ensuremath \{ \pm 1\}}
\newcommand{\rgta}{\ensuremath{\rightarrow}}

\newcommand{\R}{\mathbb R}
\newcommand{\C}{\mathbb C}

\newcommand{\Z}{\mathbb Z}

\newcommand{\pr}{\Pr}
\newcommand{\sgn}{\ensuremath{\mathds{1}^+}}
\newcommand{\dotp}[2]{\left\langle #1,#2\right\rangle}

\newcommand{\fshape}{Fourier shape}
\newcommand{\fshapes}{Fourier shapes}

\newtheorem{thm}{Theorem}[section]
\newtheorem{definition}{Definition}

\newtheorem{cor}[thm]{Corollary}
\newtheorem{lem}[thm]{Lemma}

\newtheorem*{defn}{Definition}

\newtheorem{fact}[thm]{Fact}

\newcommand{\zpm}{\{-1,0,1\}}
\newcommand{\dpm}{\pmo}

\newcommand{\dtv}{d_{TV}}
\newcommand{\hh}{\mathcal{H}}
\newcommand{\ignore}[1]{{}}
\newcommand{\bkets}[1]{\left(#1\right)}
\newcommand{\sbkets}[1]{\left[#1\right]}

\newcommand{\iprod}[2]{\langle #1,#2\rangle}
\newcommand{\nmo}[1]{\left\|#1\right\|_1}
\newcommand{\nmt}[1]{\left\|#1\right\|_2}
\newcommand{\nm}[1]{\left\|#1\right\|}
\newcommand{\nmp}[1]{\left\|#1\right\|_p}

\newcommand{\abs}[1]{\left|#1\right|}
\newcommand{\calD}{\mathcal{D}}

\newcommand{\mnote}[1]{}
\newcommand{\pnote}[1]{}
\newcommand{\dnote}[1]{}

\newcommand{\tvar}{\mathsf{Tvar}}

\date{}

\title{Pseudorandomness via the discrete Fourier transform}
\author{Parikshit Gopalan \\Microsoft Research \and Daniel M. Kane \\University of California, San Diego \and Raghu Meka \\University of California, Los Angeles}

\begin{document}
\begin{titlepage}
\thispagestyle{empty}
\maketitle
\begin{abstract}
We present a new approach to constructing unconditional pseudorandom
generators against classes of functions that involve computing a linear
function of the inputs. We give an explicit construction of a pseudorandom generator that fools the
{\em discrete Fourier transforms} of linear functions with
seed-length that is nearly logarithmic (up to polyloglog factors) in the input
size and the desired error parameter.
Our result gives a single pseudorandom generator that fools several
important classes of tests computable in logspace that have been
considered in the literature, including halfspaces (over general
domains), modular tests and combinatorial shapes. For all these
classes, our generator is the first that achieves near logarithmic
seed-length in both the input length and the error parameter. Getting such a seed-length is a natural
challenge in its own right, which needs to be overcome in order to derandomize $\rl$ --- a central question in complexity theory. %We believe it marks progress in the quest to derandomize all of {\sc   RL}.

%Our generator can be viewed as a natural and substantial generalization of the small-bias spaces of Naor and Naor \cite{NaorN93}.
Our construction combines ideas from a large body of prior work,
ranging from a classical construction of \cite{NaorN93} to the recent gradually increasing
independence paradigm  of \cite{KaneMN11, CelisRSW13,
  GopalanMRTV12}, while also introducing some novel analytic machinery
which might find other applications.
\end{abstract}

\ignore{
\begin{abstract}
The problem of constructing pseudorandom generators that fool
halfspaces has been studied intensively in recent times. For fooling
halfspaces over $\pmo^n$ with error $\epsilon$,  the best construction
known requires seed-length $O(\log(n) + \log^2(1/\epsilon))$
\cite{MekaZ13}.  For polynomially small error (specifically for $\eps
< 1/\sqrt{n}$), nothing better was known than generators  of seed-length $O(\log^2 (n))$ which fool
all of $\rl$. Getting the seed-length down to $O(\log(n/\epsilon))$ is a natural
challenge in its own right, which needs to be overcome in order to
derandomize $\rl$ \cite{MekaZ13}.

In this work we give the first pseudorandom generators for halfspaces
with nearly logarithmic seed-length. Our generators require seed-length $O(\log(n/\eps)\log\log(n/\eps))$.
We reduce the problem to fooling Fourier coefficients of
real-valued linear functions with polynomially small error. Our generator combines ideas from a classical
construction for small-bias spaces due to \cite{NaorN93} with the more recent gradually increasing
independence paradigm  of \cite{KaneMN11, CelisRSW13,
  GopalanMRTV12}, while its analysis calls for some novel analytic machinery
which might find other applications.
\end{abstract}}
\end{titlepage}

\input{intro}
\input{prelims}
\input{techlemma}
%\input{Fourier}
\input{largeNorm}
\input{alphabetredux}
\input{dimredux}

\input{together}
\input{applications}

\bibliographystyle{alpha}
\bibliography{references}

\appendix
\input{appendix}

\end{document}

%% file: intro.tex
% !TEX root = main.tex

\section{Introduction}

A central goal of computational complexity  is to understand the power that
randomness adds to efficient computation. The main questions in
this area are whether $\bpp = \ptime$ and $\rl =\logspace$, which
respectively assert that randomness can be eliminated from efficient
computation, at the price of a polynomial slowdown in time, and a
constant blowup in space. It is known that proving $\bpp = \ptime$ will
imply strong circuit lower bounds that seem out of reach of
current techniques. In contrast, proving $\rl = \logspace$,
could well be within reach. Indeed, {\em bounded-space algorithms} are a natural computational model
for which we know how to construct strong \emph{pseudo-random generators}, $\prg$s, unconditionally.

Let $\rl$ denote the class of randomized algorithms with $O(\log n)$
work space which can access the random bits in a read-once
pre-specified order. Nisan~\cite{Nisan92} devised a $\prg$ of seed
length $O(\log^2(n/\eps))$ that fools $\rl$ with error $\eps$.
This generator was subsequently used by Nisan \cite{Nisan94} to show that $\rl
\subseteq \mathsf{SC}$ and by Saks and Zhou~\cite{SaksZ99} to
prove that $\rl$ can be simulated in space $O(\log^{3/2} n)$.
Constructing $\prg$s with the optimal
$O(\log(n/\epsilon))$ seed length for this class and showing that $\rl
= \logspace$ is arguably the outstanding open problem in
derandomization (which might not require a breakthrough in lower bounds).
Despite much progress in this area
\cite{ImpagliazzoNW94,NisanZ96,RazR99,Reingold08,ReingoldTV06,BravermanRRY14,BrodyV10,KouckyNP11,De11,GopalanMRTV12},
there are few cases where we can improve on Nisan's twenty year old
bound of $O(\log^2 (n/\eps))$ \cite{Nisan92}.

\subsection{\fshapes}
A conceptual contribution of this work is to propose a class of functions in $\rl$ which we call \emph{\fshapes } that
unify and generalize the problem of fooling many natural classes of test functions that are
computable in logspace and involve computing linear combinations of
(functions of) the input variables.  In the following, let $\C_1 = \{z: |z| \leq
1\}$ be the unit-disk in the complex plane.

\begin{definition}
\label{def:fshape}
A $(m,n)$-\emph{\fshape} $f:[m]^n \to \C_1$ is a function of the form
$f(x_1,\ldots,x_n) = \prod_{j=1}^n f_j(x_j)$ where each $f_j:[m] \to
\C_1$. We refer to $m$ and $n$ as the alphabet size and the dimension
of the \fshape\, respectively.
\end{definition}

Clearly, $(m,n)$-\fshapes\, can be computed with $O(\log n)$ workspace, as long as the
bit-complexity of $\log(f_j)$ is logarithmic for each $j$; a condition that can be enforced
without loss of generality.  Since our goal is to fool functions
$f:\zo^n \to \zo$, it might be unclear why we should consider complex-valued
functions (or larger domains). The answer comes from the discrete
Fourier transform which maps integer random variables to $\C_1$.
Concretely consider a Boolean function $f:\zo^n \to \zo$ of the form
$f(x) = g(\sum_jw_jx_j)$ where $x \in \zo^n$,  $w_j \in \Z$, and $g:\Z \to \zo$ is a {\em
   simple} function like a threshold or a mod function. To fool such a 
function $f$,  it suffices to \emph{fool} the linear function $w(x) =
 \sum_jw_jx_j$. A natural way to establish the
 closeness of distributions on the integers is via the discrete
 Fourier transform. The discrete Fourier transform of $w(x)$ at
 $\alpha \in [0,1]$ is given by
\[ \phi_\alpha(w(x)) = \exp(2\pi i \alpha  \cdot w(x)) =
\prod_{j=1}^n\exp(2\pi i \alpha w_jx_j) \]
which is a Fourier shape.

Allowing a non-binary alphabet $m$ not only allows us to capture more general
classes of functions (such as combinatorial shapes), it makes the class
more robust. For instance, given a \fshape\ $f:\zo^n \to \C$,  if we
consider inputs bits in blocks of length $\log(m)$, then the resulting
function is still a Fourier shape over a larger input domain $[m]$ (in dimension $n/\log(m)$). This allows
certain compositions of $\prg$s and simplifies our construction even
for the case $m =2$.

\subsubsection{$\prg$s for \fshapes\ and their applications.}

A $\prg$ is a function $\calG:\zo^r \rgta [m]^n$. We refer to $r$ as
the seed-length of the generator. We say $\calG$ is \emph{explicit} if
the output of $\calG$ can be computed in time
$\poly(n)$.\footnote{Throughout, for a multi-set $S$, $x \in_u S$
  denotes a uniformly random element of $S$.}

\begin{definition}
\label{def:prg}
A $\prg$ $\calG:\zo^r \to [m]^n$ fools a class of
functions $\mathcal{F} = \{f:[m]^n \to \C\}$ with error $\eps$
(or $\eps$-fools $\mathcal{F}$) if for every $f \in \mathcal{F}$,
\[\left|\E_{x \in_u [m]^n}[f(x)] - \E_{y \in_u \zo^r}[f(\calG(y))]\right| < \eps.\]
%\[\left|\pr_{x \in \dpm^n}[f(x) = 1] - \pr_{y \in \pmo^r}[f(G(y)) = 1]\right| < \eps.\]
\end{definition}

We motivate the problem of constructing $\prg$s for \fshapes\ by
discussing how they capture a variety of well-studied classes like
halfspaces (over general domains), combinatorial rectangles,  modular
tests and combinatorial shapes.

\paragraph{$\prg$s for halfspaces.} 
Halfspaces are functions $h:\zo^n \to \zo$ that can be represented as
\[ h(x) = \sgn(\dotp{w}{x} - \theta)\]
for some \emph{weight} vector $w \in \Z^n$ and \emph{threshold}
$\theta \in \Z$ where $\sgn(a) = 1$ if $a \geq 0$ and $0$ otherwise.
Halfspaces are of central importance
in computational complexity, learning theory and social choice. Lower
bounds for halfspaces are trivial, whereas the problem of proving
lower bounds against depth-$2$ $\tc$ or halfspaces of
halfspaces is a frontier open problem in computational complexity. The
problem of constructing explicit $\prg$s that can fool halfspaces is a
natural challenge that has seen a lot of exciting progress recently
\cite{DGJSV09, MekaZ13, Kane11,Kane14,KothariM14}. The best known $\prg$
construction for halfspaces is that of Meka and Zuckerman
\cite{MekaZ13} who gave a $\prg$ with seed-length $O(\log n +
\log^2(1/\epsilon))$, which is $O(\log^2(n))$ for polynomially small
error. They also showed that $\prg$s against $\rl$ with
inverse polynomial error can be used to fool halfspaces, and thus
constructing better $\prg$s for halfspaces is a
necessary step towards progress for bounded-space algorithms.
However, even for special cases of halfspaces (such as derandomizing
the Chernoff bound), beating seed-length
$O(\log^2 (n))$ has proved difficult. 

We show that a $\prg$ for $(2,n)$-\fshapes\, with error $\epsilon/n^2$
also fools halfspaces with error $\epsilon$. In particular, $\prg$s fooling
\fshapes\ with polynomially small error also fool halfspaces with
small error.

\paragraph{$\prg$s for generalized halfspaces.}
$\prg$s for $(m,n)$-\fshapes\, give us $\prg$s for halfspaces not
just for  the uniform distribution over the
hypercube, but for a large class of distributions that have been
studied in the literature. We can derive these results in a unified
manner by considering the class of {\em generalized   halfspaces}.

\begin{definition}
\label{def:gh}
A generalized halfspace over $[m]^n$ is a function $g:[m]^n \to
\zo$ that can be represented as
\[ g(x) = \sgn\bkets{\sum_{j=1}^n g_j(x_j) - \theta}.\]
where $g_j:[m] \to \R$ are arbitrary functions for $j \in [n]$ and $\theta \in \R$. 
\end{definition}

$\prg$s for $(m,n)$-\fshapes\, imply $\prg$s for generalized halfspaces. This in turn captures settings
of fooling halfspaces with respect to the Gaussian distribution and
the uniform distribution on the sphere  \cite{KarninRS12,MekaZ13,Kane14,KothariM14},
and a large class of product distributions over $\R^n$
\cite{GopalanOWZ10}.

 \paragraph{Derandomizing the Chernoff-Hoeffding bound.}
 \mnote{This formulation is perhaps clean and general enough ...}
A consequence of fooling generalized halfspaces is to
{\em derandomize} Chernoff-Hoeffding type bounds for sums of
independent random variables which are ubiquitous in the analysis of
randomized algorithms. We state our result in the language of ``randomness-efficient samplers'' (cf.~\cite{Zuckerman97}). Let $X_1,\ldots,X_n$ be
independent random variables over a domain $[m]$ and let $g_1,\ldots,g_n:[m] \to [-1,1]$
be arbitrary bounded functions. The classical Chernoff-Hoeffding
bounds \cite{Hoeffding} say that  
\[ \pr\sbkets{\abs{\sum_{i=1}^n g_i(X_i) - \sum_{i=1}^n \E[g_i(X_i)]}
  \geq t} \leq 2 \exp(-t^2/4n).\] 
There has been a long line of work on showing sharp tail bounds for
pseudorandom sequences starting from \cite{SchmidtSS95} who showed
that similar tail bounds hold under limited independence. But all previous
constructions for the polynomial small error regime
required seed-length $O(\log^2(n))$. $\prg$s for generalized
halfspaces give Chernoff-Hoeffding tail bounds with polynomially
small error, with seed-length $\tilde{O}(\log(n))$.

\ignore{There
has been a long line of work on showing sharp tail bounds for
pseudorandom sequences starting from \cite{SchmidtSS95}. A $\prg$ for
halfspaces with polynomially small error implies
Chernoff-like sub-Gaussian tail bounds.  But previous constructions that guarantee such strong tail
bounds \cite{MekaZ13,SchmidtSS95} required seed-length $O(\log^2(n))$.}

\paragraph{$\prg$s for modular tests.}
An important class of functions in $\logspace$ is that of
\emph{modular tests}, i.e., functions of the form $g:\zo^n \to \zo$,
where $g(x) = \mathds{1}(\sum_i a_i x_i \bmod m \in S)$, for $m \leq M$, coefficients
$a_i \in \Z_m$ and $S \subseteq \Z_m$. Such a test is computable in
$\logspace$ as long as $M \leq \poly(n)$. The case when $m =2$
corresponds to small-bias spaces, for which optimal constructions were
first given in the seminal work of Naor and Naor \cite{NaorN93}. The
case of arbitrary $m$ was considered by \cite{LovettRTV09} (see also
\cite{MekaZ09}), their generator gives seed-length
$\tilde{O}(\log(n/\eps) + \log^2(M))$. Thus for $M = \poly(n)$, their
generator does not improve on Nisan's generator even for constant
error $\epsilon$. $\prg$s fooling $(2,n)$-\fshapes\ with polynomially small
error fools modular tests.

\ignore{
While we rigorously establish this later on, the crux here is that fooling halfspaces is equivalent to fooling all linear functions $L(x) = \sum_i w_i x_i$, where $w_i \in \Z$, in Kolmogorov or cdf distance (see Equation \eqref{eq:cdf-dist} for formal definition). Further, fooling such tests amounts to studying the discrete random variable $Y = \sum_i w_i X_i$ for $X$ distributed over $\zo^n$. On the other hand, if $Y$ is bounded in the range $[-N,N]$, $Y$ is described completely by its discrete Fourier coefficients over $\Z_N$: i.e.,
$$\hat{Y}(k) = \E[\exp(2\pi i k Y/N)],$$
where $k \in \{0,1,\ldots,N-1\}$. Finally, note that each of these Fourier coefficients is the expectation of a \fshape: For $0 \leq k \leq N-1$, define a $(2,n)$-\fshape\, $\hat{L}^k:\{0,1\}^n \to \C_1$ by
$$\hat{L}^k(x) = \prod_{j=1}^n \exp(2 \pi i k a_j x_j /N).$$
Then, $\hat{Y}(k) = \E[\hat{L}^k(X)]$. Therefore, fooling \fshapes\,
amounts to fooling the Fourier transform of the linear function
$L$. We then exploit this connection by showing that for bounded
integer-valued random variables, closeness in Kolmogorov distance is
implied by closeness in Fourier transforms of the random variables. 

Note that in the above argument to fool Halfspaces, we only needed to
fool $(2,n)$-\fshapes. While this is the case formally, considering
$(m,n)$-\fshapes\, for arbitrary $m$ plays an instrumental (and even
simplifying) role in our iterative constructions even when only trying
to fool $(2,n)$-\fshapes.}

\paragraph{$\prg$s for combinatorial shapes.}

\emph{Combinatorial shapes} were introduced in the work of
\cite{GopalanMRZ13} as a generalization of combinatorial rectangles
and to address fooling linear sums in statistical distance. These are
functions  $f:[m]^n \to \zo$ of the form
\[ f(x) = h\bkets{\sum_{i=1}^n g_i(x_i)}\]
for functions $g_i: [m] \to \zo$ and a function $h:\{0,\ldots,n\} \to
\zo$. The best previous generators of \cite{GopalanMRZ13} and
\cite{De14} for combinatorial shapes achieve a seed-length of $O(\log
(mn) + \log^2(1/\epsilon))$, $O(\log m + \log(n/\epsilon)^{3/2})$; in
particular, the best previous seed-length for polynomially small error
was $O(\log^{3/2} (n))$. $\prg$s for $(m,n)$-\fshapes\  with error $\eps/n$ imply $\prg$s for
combinatorial shapes. 

\emph{Combinatorial rectangles} are a well-studied subset of combinatorial
shapes \cite{EvenGLNV98, ArmoniSWZ96, LinialLSZ97, Lu02}.
They are functions that can be written as
$f(x) = \prod_j \mathds{1}(x_j \in A_j)$
for some arbitrary subsets $A_j \subseteq [m]$.
The best known $\prg$ due to \cite{GopalanMRTV12,GopalanY14} gives a seed-length of
$O(\log(mn/\epsilon)\log \log(mn/\epsilon))$. Combinatorial
rectangles are special cases of  \fshapes\, so our $\prg$ for $(m,n)$-\fshapes \ also fools combinatorial rectangles, 
but requires a slightly longer seed. The {\em alphabet-reduction} step in our construction
is inspired by the generator of \cite{GopalanMRTV12,GopalanY14}.

\subsubsection{Achieving optimal error dependence via \fshapes.}

We note that having generators for Fourier shapes with seed-length
$\tilde{O}(\log(n))$ even when $\eps$ is polynomially small is
essential in our reductions: we sometimes need error $\eps/\poly(n)$
for \fshapes\ in order to get $\eps$ error for our target class of
functions. Once we have this, starting with $\eps$ a sufficiently
small polynomial results in polynomially small error for the target
class of functions. 

We briefly explain why previous techniques based on {\em limit theorems} were unable to
achieve polynomially small error with optimal seed-length, by
considering the setting of halfspaces under the uniform distribution
on $\zo^n$.
Fooling halfspaces is equivalent to fooling all linear functions
$L(x) = \sum_iw_ix_i$ in Kolmogorov or cdf distance.
Previous work on fooling halfspaces
\cite{DGJSV09,MekaZ13} relies on the Berry-Ess\'een theorem, a quantiative form of the central limit theorem, to show that the cdf
of {\em regular} linear functions is close to that of the Gaussian
distribution, both under the uniform distribution and under the
pseudorandom distribution. However, even for the majority function (which is the
most regular linear function), the discreteness of $\sum_ix_i$ means
that the Kolmogorov distance from the Gaussian distribution is
$1/\sqrt{n}$, even when $x$ is uniformly random. Approaches that
show closeness in cdf distance by comparison to the Gaussian distribution seem unlikely
to give polynomially small error with optimal seed-length. 

We depart from the {\em derandomized limit
  theorem} approach taken by several previous works \cite{DGJSV09, DiakonikolasKN10, GopalanOWZ10, HarshaKM12,
  GopalanMRZ13, MekaZ13} and work directly with the Fourier transform. A crucial insight (that is formalized in
Lemma \ref{lem:introfouriertocdf}) is that fooling the Fourier
transform of linear forms to within
polynomially small error implies polynomially small Kolmogorov distance. 

%P: I think we made the below point already.
\ignore{We also remark that we work with general $(m,n)$-\fshapes\, from the beginning which gives us an advantage even when only trying to fool $(2,n)$-\fshapes\, which in turn is enough to fool halfspaces over $\zo^n$. In particular,  considering $(m,n)$-\fshapes\, for arbitrary $m$ plays an instrumental, and even simplifying role in our iterative constructions even when targeting $(2,n)$-\fshapes.}

\ignore{By working with the Fourier transform, we depart from the {\em derandomized limit
  theorem} approach taken by several previous works \cite{DGJSV09, DiakonikolasKN10, GopalanOWZ10, HarshaKM12,
  GopalanMRZ13, MekaZ13}. A crucial insight (that is formalized in
Lemma \ref{lem:introfouriertocdf}) is that fooling the Fourier
transform of linear forms to within
polynomially small error implies polynomially small Kolmogorov distance. }

\subsection{Our results}
%\mnote{Let us have $\zo$ inputs for PRG ... seems more natural somehow.}
%Parik: I'd rather not switch between 0,1 and -1,1: gives me a headache.
\ignore{
A $\prg$ is a function $\calG:\pmo^r \rgta \pmo^n$. We refer to $r$ as
the seed-length of the generator. The $\tilde{O}()$ notation hides
polylogarithmic factors in its argument. We say $\calG$ is \emph{explicit} if
the output of $\calG$ can be computed in time $\poly(n)$.

\begin{definition}
A $\prg$ $\calG:\pmo^r \to \dpm^n$ fools a class of
functions $\mathcal{F} = \{f:\dpm^n \to \R\}$ with error $\eps$
(or $\eps$-fools $\mathcal{F}$) if for every $f \in \mathcal{F}$,
\[\left|\E_{x \in \dpm^n}[f(x)] - \E_{y \in \pmo^r}[f(G(y))]\right| < \eps.\]
%\[\left|\pr_{x \in \dpm^n}[f(x) = 1] - \pr_{y \in \pmo^r}[f(G(y)) = 1]\right| < \eps.\]
\end{definition}}

Our main result is the following:

\begin{thm}
\label{thm:main}
There is an explicit generator $\calG:\zo^r \to [m]^n$ that fools all $(m,n)$-\fshapes\, with error $\eps$, and has seed-length $r = O(\log(mn/\epsilon) \cdot (\log \log (mn/\epsilon))^2)$.
\end{thm}

We now state various corollaries of our main result starting with
fooling halfspaces. 

\begin{cor}
\label{cor:halfspaces}
There is an explicit generator $\calG:\zo^r \to \zo^n$ that fools
halfspaces over $\zo^n$ under the uniform distribution with error
$\eps$, and has seed-length $r = O(\log(n/\epsilon)(\log \log
(n/\epsilon))^2)$. 
\end{cor}

The best previous generator due to \cite{MekaZ13} had a seed-length of $O(\log n +
\log^2(1/\epsilon))$, which is $O(\log^2 n)$ for polynomially small error $\epsilon$.

We also get a $\prg$ with similar parameters for generalized halfspaces.
\begin{cor}\label{cor:halfspaces2}
There is an explicit generator $\calG:\zo^r \to [m]^n$ that
$\eps$-fools {\em generalized halfspaces} over $[m]^n$, and has 
seed-length $r = O(\log(mn/\epsilon) \cdot (\log \log (mn/\epsilon))^2)$. 
\end{cor}

From this we can derive  $\prg$s with seed-length
$O(\log(n/\eps)(\log\log(n/\eps))^2)$ for fooling halfspaces with error $\eps$ under the Gaussian distribution and the uniform distribution on
the sphere. Indeed, we get the following bound for arbitrary product
distributions over $\R^n$, which depends on the $4^{th}$ moment of
each co-ordinate. 

\begin{cor}
\label{cor:general}
Let $X$ be a product distribution on $\R^n$ such that for all $i \in [n]$,
\[ \E[X_i] =0, \E[X_i^2] =1, \E[X_i^4] \leq C. \]
There exists an explicit generator $\calG:\zo^r \to \R^n$ such that if
$Y = \calG(z)$, then for every halfspace $h:\R^n \to \zo$, 
\[ \abs{\E[h(X)] - \E[h(Y)]} \leq \eps.\]
The generator $G$ has seed-length $r = O(\log(nC/\eps)(\log\log(nC/\eps))^2)$. 
\end{cor}
This improves on the result of \cite{GopalanOWZ10} who obtained
seedlength $O(\log(nC/\eps)\log(C/\eps)$ for this setting via a suitable
modification of the generator from \cite{MekaZ13}.

The next corollary is a near-optimal derandomization of the Chernoff-Hoeffding bounds.
To get a similar guarantee, the best known seed-length that follows
from previous work \cite{SchmidtSS95,MekaZ13, GopalanOWZ10} was
$O(\log (mn) + \log^2(1/\epsilon))$.  

\begin{cor}
\label{cor:ch-bound}
Let $X_1,\ldots,X_n$ be independent random variables over the domain
$[m]$. Let $g_1,\ldots,g_n:[m] \to [-1,1]$ be arbitrary bounded
functions. There exists an explicit generator $G:\zo^r \to [m]^n$ 
such that if $(Y_1,\ldots,Y_n) = G(z)$ where $z
\in_u \zo^r$, then $Y_i$ is distributed identically to $X_i$ and
\[ \pr\sbkets{\abs{\sum_{i=1}^n g_i(Y_i) - \sum_{i=1}^n \E[g_i(Y_i)]}
  \geq t} \leq 2 \exp(-t^2/2n) + \epsilon.\]
$G$ has seed-length $r = O(\log(mn/\epsilon)(\log \log
(mn/\epsilon))^2)$. 
\end{cor}

We get the first generator for fooling modular tests whose dependence
on the modulus $M$ is near-logarithmic. The best previous generator
from \cite{LovettRTV09} had a seed-length of $\tilde{O}(\log(n/\eps) +
\log^2(M))$, which is $\tilde{O}(\log^2 n)$ for $M = \poly(n)$.

\begin{cor}
\label{cor:modular}
There is an explicit generator $\calG:\zo^r \to \zo^n$ that fools
all linear tests modulo $m$ for all $m \leq M$ with error $\eps$, and has seed-length $r =
O(\log(Mn/\eps)\cdot (\log\log(Mn/\eps))^2)$.
\end{cor}

Finally, we get a generator with near-logarithmic seedlength for
fooling combinatorial shapes. 
\cite{GopalanMRZ13} gave a $\prg$ for combinatorial shapes
with a seed-length of $O(\log(mn) + \log^2(1/\epsilon))$. This was
improved recently by De \cite{De14} who gave a $\prg$ with seed-length
$O(\log m + \log(n/\epsilon)^{3/2})$; in particular, the best previous
seed-length for polynomially small error was $O((\log(n)^{3/2})$.

\begin{cor}
\label{cor:cshapes}
There is an explicit generator $\calG:\zo^r \to [m]^n$ that fools
$(m,n)$-combinatorial shapes to error $\eps$ and has seed-length $r =
O(\log(mn/\epsilon)(\log \log (mn/\epsilon))^2)$.
\end{cor}

%, which does not rely on results about expansion in Lie groups.

\subsection{Other related work}\label{sec:previouswork}

Starting with the work of Diakonikolas et al.~\cite{DGJSV09}, there
has been a lot of interest in
constructing $\prg$s for halfspaces and related classes such as intersections
of halfspaces and polynomial threshold functions over the domain $\pmo^n$
\cite{DiakonikolasKN10, GopalanOWZ10, HarshaKM12, MekaZ13, Kane11,
  Kane11b,Kane14}. Rabani and Shpilka \cite{RabaniS09} construct optimal
hitting set generators for halfspaces over $\pmo^n$; hitting set
generators are weaker than $\prg$s.

Another line of work gives $\prg$s for halfspaces for the uniform distribution over
the sphere (\emph{spherical caps}) or the Gaussian distribution. For spherical caps, Karnin, Rabani and Shpilka \cite{KarninRS12} gave a $\prg$ with a seed-length of $O(\log n + \log^2(1/\eps))$. For the Gaussian distribution, \cite{Kane14} gave a $\prg$ which achieves a seed-length of $O(\log n + \log^{3/2}(1/\eps))$. Recently, \cite{KothariM14} gave the first $\prg$s for these settings with seedlength $O((\log (n/\eps))(\log\log(n/\eps)))$. Fooling halfspaces over the hypercube is known to be {\em harder} than the Gaussian setting or the uniform distribution on the sphere; hence our result gives a construction with similar parameters up to a $O(\log \log n)$ factor. At a high level, \cite{KothariM14} also uses a iterative dimension reduction approach like in \cite{KaneMN11, CelisRSW13, GopalanMRTV12}; however, the final construction and its analysis are significantly different from ours.

Gopalan {\em et al.} \cite{GopalanOWZ10} gave a generator fooling halfspaces under
product distributions with bounded fourth moments, whose
seed-length is $O(\log(n/\eps)\log(1/\eps))$. 

The present work completely subsumes a manuscript of the authors which
essentially solved the special-case of derandomizing Chernoff bounds
and a special class of halfspaces \cite{GopalanKM14a}.  

\ignore{
Fooling halfspaces over the hypercube is known to be {\em harder} than the Gaussian setting or the uniform distribution on the sphere. Recently, \cite{KothariM14} gave the first $\prg$s for these settings with seedlength $O((\log (n/\eps))(\log\log(n/\eps)))$. Our result gives a construction with similar parameters up to a $O(\log \log n)$ factor. Gopalan {\em
  et al.} \cite{GopalanOWZ10} gave a generator fooling halfspaces under
product distributions where the bias is bounded away from $0$, whose
seed-length is $O(\log(n/\eps)\log(1/\eps))$. We improve on their
seed-length which  while also relaxing the bias constraint.

Another line of work gives $\prg$s for halfspaces for the uniform distribution over
the sphere (\emph{spherical caps}) or the Gaussian distribution.
For spherical caps, Karnin, Rabani and Shpilka \cite{KarninRS12} gave
a $\prg$ with a seed-length of $O(\log n + \log^2(1/\eps))$. For the
Gaussian distribution, \cite{Kane14} gave a $\prg$ which achieves a
seed-length of $O(\log n + \log^{3/2}(1/\eps))$. Recently, Kothari and
Meka \cite{KothariM14} gave a $\prg$ for spherical caps and the
Gaussian distribution with seed-length $\tilde{O}(\log(n/\eps))$. At a
high level, \cite{KothariM14} also uses a iterative dimension
reduction approach like in \cite{KaneMN11, CelisRSW13, GopalanMRTV12};
however, the final construction and its analysis are significantly
different from ours.}

%\emph{Combinatorial shapes} were introduced by \cite{GopalanMRZ13}. De \cite{De14} recently gave a $\prg$ for combinatorial shapes having seed-length $O((\log(n/\eps))^{3/2})$, improving on the $O(\log(n/\eps))^2)$ bound ffrom \cite{GopalanMRZ13}.

\newcommand{\Par}{\ensuremath {\mathsf{Parity}}}

\section{Proof overview}\label{sec:overview}
We describe our $\prg$ for \fshapes\, as in Theorem
\ref{thm:main}. The various corollaries are derived from this Theorem using properties of the discrete Fourier transform of
integer-valued random variables.

Let us first consider a very  simple $\prg$: $O(1)$-wise independent
distributions over $[m]^n$. At a glance, it appears to do very
poorly as it is easy to express the parity of a subset of bits as a \fshape\, and parities are not fooled even by $(n-1)$-wise independence. %Bounded independence is only sufficient to fool parities where $|S| \leq k$.
The starting point for our construction is that bounded independence
does fool a special but important class of \fshapes,  namely those with
polynomially small {\em total variance}.

For a complex valued random variable $Z$, define
the variance of $Z$ as
\[ \sigma^2(Z) = \E\sbkets{|Z - \E[Z]|^2} = \E[|Z|^2] - |\E[Z]|^2.\]
It is easy to verify that
\[ \sigma^2(Z) + |\E[Z]|^2 = E[|Z|^2], \]
so that if $Z$ takes values in $\C_1$, then
\[ \sigma^2(Z) + |\E[Z]|^2 \leq 1. \]

The \emph{total-variance} of a $(m,n)$-\fshape\, $f:[m]^n \to \C_1$ with $f(x) = \prod_{j=1}^n f_j(x_j)$ is defined as
\[ \tvar(f) = \sum_j \sigma^2(f_j(x_j)).\]
To gain some intuition for why this is a natural quantity,
note that $\tvar(f)$ gives an easy upper bound on the expectation of a \fshape:
\begin{equation}
\label{eq:intro1}
\abs{\E_{x \in [m]^n}[f(x)]} = \prod_j \abs{\E[f_j(x_j)]} \leq
\prod_j\sqrt{1 - \sigma^2(f_j(x_j))}  \leq \exp(-\tvar(f)/2).
\end{equation}

This inequality suggests a natural dichotomy for the task of fooling
\fshapes.  It suggests that high variance shapes where $\tvar(f)
\gg \log(1/\eps)$ are {\em easy} in the sense that $\E[f] \ll
\eps$ is small for such \fshapes. So a $\prg$ for such shapes only needs to
ensure that $\E[f]$ is also sufficiently small under the pseudorandom output. %The main challenge is to match the expectation in the low-variance case where $\tvar(f)$ is smaller than $O(\log(1/\eps))$.

To complement the above, we show that if the total-variance $\tvar(f)$ is very small, then generators based on limited independence do fairly well. Concretely, our main technical
lemma says that limited independence fools products of bounded (complex-valued) random
variables, provided that the sum of their variances is small.
\mnote{I think it is better to have the previous version here with the exp(O(k)) floating around.}
\begin{lem}\label{intro:maintech}
Let $Y_1,\ldots,Y_n$ be $k$-wise independent random variables taking
values in $\C_1$. Then, % there exists constant $a_1, a_2$ such that for $k \geq a_1$,
\[ \abs{\E[Y_1\cdots Y_n] - \prod_{j=1}^n \E[Y_j]} \leq \exp(O(k))
\bkets{\frac{\sum_{j=1}^n \sigma^2(Y_j)}{\sqrt{k}}}^{\Omega(k)}.\]
\end{lem}

We defer discussion of the proof to Section
\ref{sec:coretechintro}, and continue the description of our $\prg$
construction.  Recall that we are trying to fool a $(m,n)$-\fshape\,
$f:[m]^n \to \C_1$ with $\tvar(f) \leq O(\log(1/\eps)$ to error $\epsilon=\poly(1/nm)$. It is helpful to
think of the desired error $\epsilon$ as being fixed at the beginning
and staying unchanged through our iterations, while
$m$ and $n$ change during the iterations. Generating $k$-wise independent distributions over $[m]^n$ takes
$O(k \log(mn))$ random bits.  Thus if we use $k = O(\log(1/\eps))$-wise
independence, we would achieve error $\eps$, but with seed-length
$O(\log(1/\eps)\log(mn))$ rather than $O(\log(1/\eps))$.

\mnote{This paragraph has changed significantly ...}
On the other hand, if $\tvar(f) \leq 1/(mn)^c$ for a fixed constant $c$, then choosing $k = O(\log(1/\eps)/(\log mn))$-wise independence is enough to get error $\epsilon$ while also achieving seed-length $O(k\log(mn)) = O(\log(1/\epsilon))$ as desired. We exploit this observation by combining the use of limited independence with the recent {\em iterative-dimension-reduction} paradigm of \cite{KaneMN11, CelisRSW13, GopalanMRTV12}. Our construction reduces the problem of fooling \fshapes\, with $\tvar(f) \leq O(\log(1/\eps))$ through a sequence of iterations to fooling \fshapes\ where the total variance is polynomially small in $m, n$ in each iteration and then uses limited independence in each iteration.

\ignore{To reduce the seed-length, we combine the use of limited independence with the recent gradually increasing independence paradigm  of \cite{KaneMN11, CelisRSW13,
  GopalanMRTV12}. Our goal is to use seed-length
$O(\log(1/\eps))$  per iteration, which only allows for $k =
\log(1/\eps)/\log(mn)$-wise independence.
Lemma \ref{intro:maintech} implies that such  distributions over fool $f$ with error at most
$(\tvar(f))^{-\Omega(\log(1/\eps)/\log(mn))}$. This is bounded by $\eps^{O(1)}$ for
if $\tvar(f)  =O(1/(mn)^c)$ for some constant $c >0$. Our construction reduces the problem of
fooling \fshapes\ with $\tvar(f) \leq O(\log(1/\eps))$ to a sequence of
iterations, each fooling \fshapes\ where the total variance is polynomially small in $m$ and
$n$ in each iteration and then using limited independence.}

To conclude our high-level description, our generator consists of three modular parts. The
first is a generator for \fshapes\ with high total variance: $\tvar(f)
\geq  \poly(\log(1/\epsilon))$. We then give two reductions to handle
low variance \fshapes: an alphabet-reduction step reduces the alphabet $m$ down to $\sqrt{m}$
and leaves $n$ unchanged, and a dimension-reduction step that reduces
the dimension from $n$ to $\sqrt{n}$ while possibly blowing up the
alphabet to $\poly(1/\eps)$. We describe each of these parts in more
detail below.

\subsection{Fooling high-variance \fshapes}
We construct a $\prg$ with seed-length $O(\log(mn/\epsilon) \log
\log(1/\epsilon))$ which $\epsilon$-fools $(m,n)$-\fshapes\, $f$ when
$\tvar(f) \geq (\log(1/\epsilon))^{C}$ for some sufficiently large
constant $C$.  We build the generator in two steps.

In the first step, we build a $\prg$ with seed-length $O(\log(mn))$
which achieves constant error  for $(m,n)$-\fshapes\, $f$
with $\tvar(f) \geq 1$. In the second step, we drive the error down to
$\epsilon$ as follows. We hash the coordinates into roughly
$(\log(1/\epsilon))^{O(1)}$ buckets, so that for at least
$\Omega(\log(1/\epsilon))$ buckets, $f$ restricted to the coordinates
within the bucket has total-variance at least $1$. We use the
$\prg$ with constant error within each bucket, while the seeds
across buckets are recycled using a $\prg$ for small-space
algorithms. This construction is inspired by the construction of
small-bias spaces due to Naor and Naor \cite{NaorN93}; the difference
being that we use generators for space bounded algorithms for
amplification, as opposed to expander random walks as done in \cite{NaorN93}.

%The above gives us a handle on high-variance \fshapes. To handle general \fshapes, we simplify the problem by series of reductions. We will exploit this observation by going through a sequence of reductions which gradually reduce the total-variance of \fshapes\, we have to study. The first step in this direction is the following which says that we can always assume that the \emph{alphabet-size} $m$ is polynomially bounded in terms of the \emph{dimension} $n$.

\subsection{Alphabet-reduction}
The next building block in our construction is
\emph{alphabet-reduction} which helps us assume without loss of
generality that the alphabet-size $m$ is polynomially bounded in terms
of the dimension $n$. This is motivated by the construction of
\cite{GopalanMRTV12}.

Concretely, we show that constructing an $\epsilon$-$\prg$ for
$(m,n)$-\fshapes\, can be reduced to that of constructing an
$\epsilon'$-$\prg$ for $(n^4,n)$-\fshapes\, for $\epsilon' \approx
\epsilon/(\log m)$. The alphabet-reduction step consists of $(\log \log m)$ steps where in
each step we reduce fooling $(m,n)$-\fshapes\, for $m > n^4$, to that
of fooling $(\sqrt{m},n)$-\fshapes, at the cost of
$O(\log(m/\epsilon))$ random bits.

We now describe a single step that reduces the alphabet from $m$ to $\sqrt{m}$. Consider the following procedure
for generating a uniformly random element in $[m]^n$:
\begin{itemize}
\item For $D \approx \sqrt{m}$, sample uniformly random subsets
\[ S_1 = \{X[1,1],X[1,2],\ldots,X[D,1]\},\ldots,S_n = \{X[1,n], X[2,n],\ldots,X[D,n]\} \subseteq [m].\]
\item Sample $Y = (Y_1,\ldots,Y_n)$ uniformly at random from $[D]^n$.
\item Output $(Z_1,\ldots,Z_n)$, where $Z_j = X[Y_j,j]$.
\end{itemize}
Our goal is to derandomize this procedure.
The key observation is that once the subsets $S_1,\ldots,S_n$ are
chosen, we are left with a $(D,n)$-\fshape\ as a function of
$Y$. So the choice of $Y$ can be derandomized using a $\prg$ for
\fshapes\ with  alphabet $[D]$, and  it suffices to derandomize the choice of the
$X$'s. A calculation shows that (because the $Y$'s are uniformly
random), derandomizing the choice of the $X$'s reduces to that of
fooling a \fshape\, of total-variance $1/m^{\Omega(1)}$. Lemma
\ref{intro:maintech} implies that this can be done with limited
independence.

\subsection{Dimension-reduction for low-variance \fshapes}

We show  that constructing an $\epsilon$-$\prg$ for $(n^4,
n)$-\fshapes\, $f$ with $\tvar(f) \leq \poly(\log(mn/\epsilon))$ can
be reduced to that of $\epsilon'$-fooling
$(\poly(n/\epsilon),\sqrt{n})$-\fshapes\, for $\epsilon' \approx
\epsilon/\log n$. Note that here we decreased the dimension at the
expense of increasing the alphabet-size. However, this can be fixed by
employing another iteration of alphabet-reduction. This is the
reason why considering $(m,n)$-\fshapes\, for arbitrary $m$ helps us
even if we were only trying to fool $(2,n)$-\fshapes. The dimension-reduction proceeds as follows:
\begin{enumerate}
\item We first hash the coordinates into roughly $\sqrt{n}$ buckets
  using a $k$-wise independent hash function $h \in_u \hh = \{h:[n]
  \leftarrow [\sqrt{n}]\}$ for $k \approx O(\log(n/\epsilon)/\log
  n)$. Note that this only requires $O(\log(n/\epsilon))$ random
  bits.
\item For the coordinates within each bucket we use a
  $k'$-wise independent string in $[m]^n$ for $k' \approx
  O(\log(n/\epsilon)/\log n)$. We use true independence across
  buckets. Note that this requires $\sqrt{n}$ independent seeds of length $r = O(\log(n/\epsilon))$.
\end{enumerate}
While the above process requires too many random bits by itself, it
is easy to analyze. We then reduce the seed-length by observing that
if we fix the hash function $h$, then what we are left with as a function of the
seeds used for generating the symbols in each bucket is a
$(2^r \leq \poly(n/\eps),\sqrt{n})$-\fshape. So rather than using
independent seeds, we can use the output of a generator for such \fshapes.

The analysis of the above construction again relies on Lemma
\ref{intro:maintech}. The intuition is that since $\tvar(f) \leq
\poly(\log(n/\epsilon))$, and we are hashing into $\sqrt{n}$ buckets,
for most hash functions $h$ the \fshape\, restricted to each bucket
has variance  $O(1/n^c)$ for some fixed constant $c > 0$. By Lemma \ref{intro:maintech}, limited
independence fools such \fshapes.

\subsection{Main Technical Lemma}\label{sec:coretechintro}

The lemma can be seen as a generalization of a similar result proved
for real-valued random variables in \cite{GopalanY14}(who also have an
additional restriction on the means of the random variables
$Y_j$). However, the generalization to complex-valued variables is
substantial and seems to require different proof techniques.

We first consider the case where the $Y_j$'s not
only have small total-variance, but also have small
absolute deviation from their means. Concretely, let $Y_j = \mu_j(1 +
Z_j)$ where $\E[Z_j] = 0$ and $|Z_j| \leq 1/2$. In this case, we do a
variable change $W_j = \log(1+Z_j)$ (taking the principal branch of
the algorithm) to rewrite $$\prod_j Y_j = \prod_j \mu_j (1+ Z_j) = \prod_j \mu_j \cdot \exp\bkets{\sum_j W_j}.$$
We then argue that $\exp(\sum_j W_j)$ can be approximated by a
polynomial $P(W_1,\ldots,W_n)$ of degree less than $k$ with small expected error. The polynomial $P$ is obtained by truncating the Taylor series expansion of the $\exp(\;)$ function. Once, we have such a low-degree polynomial approximator, the claim follows as limited independence fools low-degree polynomials.

To handle the general case  where $Z_j$'s are not necessarily bounded, we use an inclusion-exclusion argument and exploit the fact that with high probability, not many of the $Z_j$'s (say more than $k/2$) will deviate too much from their expectation. We leave the details to the actual proof.
%%% Local Variables:
%%% mode: latex
%%% TeX-master: "main"
%%% End:

%% file: prelims.tex
% !TEX root = main.tex
\section{Preliminaries}
\label{sec:prelims}

We start with some notation:
\begin{itemize}
%\item Define the function $e:\R\rightarrow\C$ by $e(x)=\exp(2\pi i x)$.
%\item For $v \in \R^n$, define $\phi_{v}(x) = \exp(2 \pi i (v \cdot x))$.
\item For $v \in\R^n$ and a hash function $h:[n] \to [m]$, define
\begin{align}
\label{eq:hv}
h(v) = \sum_{j=1}^m \|v_{|h^{-1}(j)}\|_2^4
\end{align}
\item $\C_1 = \{z: z \in \C, |z| \leq 1\}$ be the unit disk in the complex plane.
\item For a complex valued random variable $Z$,
$$Var(Z) \equiv \sigma^2(Z) \equiv \E\sbkets{|Z - \E[Z]|^2}.$$
\item Unless otherwise stated $c,C$ denote universal constants.
\item Throughout we assume that $n$ is sufficiently large and that
$\delta,\epsilon>0$ are sufficiently small.
\item For positive functions $f,g,h$ we write $f = g+ O(h)$ when $|f-g| = O(h)$. 
\item For a integer-valued random variable $Z$, its Fourier transform is given as follows: for $\alpha \in [0,1]$, $\hat{Z}(\alpha) = \E[\exp(2 \pi i \alpha Z)]$. Further, given the Fourier coefficients $\hat{Z}(\alpha)$, one can compute the probability density function of $Z$ as follows: for any integer $j$,
$$\pr[Z = j] = \int_{0}^1 \exp(2\pi i j \alpha) \hat{Z}(\alpha) \,d \alpha.$$
\end{itemize}

\begin{defn}
For $n,m,\delta>0$ we say that a family of hash functions $\hh = \{h:[n] \to [m]\}$ is $\delta$-biased if for any $r \leq n$ distinct indices $i_1,i_2,\ldots,i_r \in [n]$ and $j_1,\ldots,j_r \in [m]$,
\[ \pr_{h\in_u \hh} \sbkets{ h(i_1) = j_1 \,\wedge \,h(i_2) = j_2
  \,\wedge\,\cdots\,\wedge h(i_r) = j_r} = \frac{1}{m^r} \pm \delta. \]

We say that such a family is $k$-wise independent if the above holds with $\delta=0$ for all $r\leq k$.

We say that a distribution over $\dpm^n$ is $\delta$-biased or $k$-wise independent if the corresponding family of functions $h:[n]\to[2]$ is.
\end{defn}

Such families of functions can be generated efficiently using small seeds.
\begin{fact}\label{HashFamilyFact}
For $n,m,k,\delta>0$, there exist explicit $\delta$-biased families of hash functions $\hh = \{h:[n]\to [m]\}$ that can be generated efficiently from a seed of length $s=O(\log(n/\delta))$. There are also, explicit $k$-wise independent families that can be generated efficiently from a seed of length $s=O(k\log(nm))$.

Taking the pointwise sum of such generators modulo $m$ gives a family of hash functions that is both $\delta$-biased and $k$-wise independent generated from a seed of length $s=O(\log(n/\delta)+k\log(nm))$.
\end{fact}

\subsection{Basic Results}

We start with the simple observation that to $\delta$-fool an $(m,n)$-\fshape\,$f$, we can assume the functions in $f$ have bit-precision $2\log_2(n/\delta)$. This observation will be useful when we use PRGs for small-space machines to fool \fshapes\, in certain parameter regimes. 
\begin{lem}\label{lem:bounded}
If a $\prg$ $\calG:\zo^r \to [m]^n$ $\delta$-fools $(m,n)$-\fshapes\, $f = \prod_j f_j$ when $\log(f_j)$'s have bit precision $2\log_2(n/\delta)$, then $\calG$ fools all $(m,n)$-\fshapes\, with error at most $2\delta$.
\end{lem}
\begin{proof}
Consider an arbitrary $(m,n)$-\fshape\, $f:[m]^n \to \C_1$ with $f = \prod_j f_j$. Let $\tilde{f_j}:[m] \to \C_1$ be obtained by truncating the $\log(f_j)$'s to $2\log_2(n/\delta)$ bits. Then, $|f_j(x_j) - \tilde{f_j}(x_j)| \leq \delta/n$ for all $x_j \in [m]$. Therefore, if we define $\tilde{f} = \prod_j \tilde{f_j}$, then for any $x \in [m]^n$, (as the $f_j$'s and $\tilde{f}_j$'s are in $\C_1$)
$$\abs{f(x) - \tilde{f}(x)} = \abs{\prod_j f_j(x_j) - \prod_j \tilde{f}(x_j)} \leq \sum_j \abs{f_j(x_j) - \tilde{f}(x_j)} \leq \delta. $$
The claim now follows as the above inequality holds point-wise and by assumption, $\calG$ $\delta$-fools $\tilde{f}$.
\end{proof}

We collect some known results about pseudorandomness and prove some
other technical results that will be used later.

\begin{comment}
We will use the following result from \cite{GopalanMRZ13} giving $\prg$s
for signed majorities.
\begin{thm}\cite{GopalanMRZ13}
\label{simpleGeneratorThm}
For $n,\eps>0$ there exists an explicit pseudorandom generator, $Y\in \dpm^n$ generated from a seed of length $s=O(\log(n)+\log^2(1/\epsilon))$ so that for any $v\in\zpm^n$ and $X \in_u \dpm^n$, we have that $\dtv(v\cdot Y,v\cdot X)\leq \epsilon$.
\end{thm}
\end{comment}

We shall use $\prg$s for small-space machines or read-once branching programs (ROBP) of Nisan \cite{Nisan92}, \cite{NisanZ96} and Impagliazzo, Nisan and Wigderson \cite{ImpagliazzoNW94}. We extend the usual definitions of read-once branching programs to compute complex-valued functions; the results of \cite{Nisan92}, \cite{NisanZ96}, \cite{ImpagliazzoNW94} apply to this extended model readily\footnote{This is because these results in fact give guarantees in terms of statistical distance.}. 
\begin{definition}[$(S,D,T)$-ROBP]
An $(S,D,T)$-ROBP $M$ is a layered directed graph with $T+1$ layers and $2^S$ vertices per layer with the following properties.
\begin{itemize}
\item The first layer has a single \emph{start} node and the vertices in the last layer are labeled by complex numbers from $\C_1$.%last layer has two nodes labeled $0,1$ respectively.
\item A vertex $v$ in layer $i$, $0\leq i < T$ has $2^D$ edges to layer $i+1$ each labeled with an element of $\zo^D$.
\end{itemize}
A graph $M$ as above naturally defines a function $M:\left(\zo^D\right)^T \to \C_1$ where on input $(z_1,\ldots,z_T) \in \left(\zo^D\right)^T$ one traverses the edges of the graph according to the labels $z_1,\ldots,z_T$ and outputs the label of the final vertex reached.
\end{definition}

\begin{thm}[\cite{Nisan92}, \cite{ImpagliazzoNW94}]\label{th:inwprg}
There exists an explicit $\prg$ $\calG^{INW}:\zo^r \to \left(\zo^D\right)^T$
which $\epsilon$-fools $(S,D,T)$-branching programs and has
seed-length $r = O(D + S \log T + \log(T/\delta) \cdot (\log T))$.
\end{thm}

\begin{thm}[\cite{NisanZ96}]\label{th:nz}
For all $C > 1$ and $0 < c < 1$, there exists an explicit $\prg$ $\calG^{NZ}:\zo^r \to \left(\zo^D\right)^T$
which $\epsilon$-fools $(S,S,S^C)$-branching programs for $\epsilon = 2^{-\log^{1-c}S}$ and has
seed-length $r = O(S)$.
\end{thm}

The next two lemmas quantify load-balancing properties of
$\delta$-biased hash functions in terms of the $\ell_p$-norms of
vectors. Proofs can be found in Appendix \ref{sec:appendix}.

\begin{lem}\label{hashmomentsLem}
Let $p\geq 2$ be an integer. Let $v \in \R^n$ and $\hh = \{h:[n] \to [m]\}$ be either a $\delta$-biased hash family for $\delta>0$ or a $p$-wise independent family for $\delta=0$. Then
$$\E[h(v)^p] \leq O(p)^{2p} \left(\frac{\|v\|_2^4}{m}\right)^p + O(p)^{2p} \|v\|_4^{4p}+ m^p \|v\|_2^{4p}\delta.$$
\end{lem}

\begin{lem}\label{lm:hashing}
For all $v \in \R^n_+$, let $p\geq 2$ be even and $\hh = \{h:[n] \to [m]\}$ a $p$-wise independent family, and $j \in [m]$,
$$\pr\sbkets{ \abs{\nmo{v_{|h^{-1}(j)}} - \nmo{v}/m} \geq t} \leq \frac{O( p)^{p/2} \nmt{v}^p }{t^p}.$$
%$$\pr_{h \in_u \hh} \sbkets{\left|\nmo{q_{|h^{-1}(j)}} - \frac{\nmo{q}}{m}\right| \geq \frac{c_\alpha\sqrt{\log(1/\delta)}}{\nms{q}^\alpha} \cdot \nmo{q}} \leq \delta.$$
\end{lem}

%%% Local Variables:
%%% mode: latex
%%% TeX-master: "main"
%%% End:

%% file: techlemma.tex
% !TEX root = main.tex
\section{Fooling products of low-variance random variables}\label{sec:maintech} 
We now show one of our main technical claims that products of complex-valued random variables are fooled by limited independence if the sum of variances of the random variables is small. The lemma is essentially equivalent to saying that limited independence fools low-variance \fshapes.
\begin{lem}\label{lm:maintech}
Let $Y_1,\ldots,Y_n$ be $k$-wise independent random variables taking values in $\C_1$. Then,
$$\abs{\E[Y_1\cdots Y_n] - \prod_{j=1}^n \E[Y_j]} \leq \exp(O(k))\cdot \bkets{\frac{\sum_j \sigma^2(Y_j)}{k}}^{\Omega(k)}.$$
\end{lem}

More concretely, let $X_1,\ldots,X_n$ be independent random variables taking values in $\C_1$. Let $\sigma_i^2 = \var(X_i)$ and $\sum_{i=1}^n \sigma^2_i\leq \sigma^2$. Let $k$ be a positive even integer and let $Y_1,\ldots,Y_n$ be a $Ck$-wise independent family of random variables with each $Y_i$ distributed identically to $X_i$. Then, we will show that for $C$ a sufficiently big constant, 
\begin{equation}\label{mainEqn}
\left|\E\sbkets{Y_1 \cdots Y_n} - \E\sbkets{X_1 \cdots X_n}  \right| = \exp(O(k)) \cdot (\sigma/\sqrt{k})^k.
\end{equation}

We start with the following standard bound on moments of bounded random variables whose proof is deferred to appendix \ref{app:maintech}. 
\begin{lem}\label{lm:mombound}
Let $Z_1,\ldots,Z_n \in \C$ be random variables with $\E[Z_i] = 0$, $\|Z_i\|_\infty < B$ and $\sum_i Var(Z_i) \leq \sigma^2$. Then, for all even positive integers $k$,
$$\E\sbkets{\abs{\sum_i Z_i}^k} \leq 2^{O(k)} (\sigma \sqrt{k} + Bk)^k.$$
\end{lem}

We also use some elementary properties of the (complex-valued) log and exponential functions:
\begin{lem}\label{lm:logexpcomplex}
\begin{enumerate}
\item For $z \in \C$ with $|z| \leq 1/2$, $|\log(1+z)| \leq 2|z|$, where we take the principle branch of the logarithm.
\item For $w \in \C$ and $k > 0$,
$$\abs{\exp(w) - \sum_{j=0}^{k-1} w^k/k!}\leq O(1) \frac{|w|^k}{k!}\cdot \max(1,\exp(\Re(w))).$$
\item For a random variable $Z \in \C$ with $|Z|_\infty \leq 1/2$, $\E[Z]=0$, and  $W = \log(1+Z)$ the principle branch of the logarithm function (phase between $(-\pi,\pi)$), $Var(W) \leq 4 Var(Z)$.
\item For any complex-valued random variable $W \in \C$, $|\exp(\E[W])| \leq \E[|\exp(W)|]$.
\end{enumerate}
\end{lem}
\begin{proof}
Claims (1), (2) follow from the Taylor series expansions for the complex-valued log and exponential functions.

For (3), note that $Var(W) \leq \E[|W|^2] \leq 4 \E[|Z|^2] = 4 Var(Z)$.

For (4), note that $|\exp(\E[W])| = |\exp(E[\Re(W)])|$ and similarly $|\exp(W)| = |\exp(\Re(W))|$. The statement now follows from Jensen's inequality applied to the random variable $\Re(W)$.
\end{proof}

We prove Lemma \ref{lm:maintech} or equivalently, Equation \eqref{mainEqn} by proving a sequence of increasingly stronger claims. We begin by proving that Equation \eqref{mainEqn} holds if $X_j$'s have small absolute deviation, i.e., lie in a disk of small radius about a fixed point.
\begin{lem}\label{logLem}
Let $X_i$ and $Y_i$ be as above. Furthermore, assume that $Y_i = \mu_i(1+Z_i)$ for complex numbers $\mu_i = \E[Y_i]$ and random variables $Z_i$ so that with probability $1$, $|Z_i|\leq B\leq 1/2$ for all $i$. Let $\tilde \sigma_i^2=\var(Z_i)$, and $\tilde\sigma^2 =\sum_{i=1}^n \tilde \sigma_i^2$. Then we have that
$$
\left|\E\left[X_1 \cdots X_n \right] - \E\left[Y_1 \cdots Y_n\right]  \right| = \exp(O(k)) \cdot (\tilde\sigma/k^{1/2}+B)^k.
$$
\end{lem}
\begin{proof}
Let $W_j=\log(1+Z_j)$, taking the principle branch of the logarithm function and let $W_j' = W_j - \E[W_j]$. Then, by Lemma \ref{lm:logexpcomplex} (1), (3), $|W_j| \leq 2 |Z_j| \leq 2B$, so that $|W_i|' \leq 4B$ and $\var(W_j') = O(\tilde\sigma_j^2)$. Finally, let $W = \sum_{j=1}^n W_j'$.

Now, by Lemma \ref{lm:logexpcomplex} (3)
\begin{align*}
\prod_{i=1}^n Y_i &= \prod_{i=1}^n \left(\mu_i\exp(\E[W_i])\right) \exp(W)\\
&= \prod_{i=1}^n \left(\mu_i\exp(\E[W_i])\right) \left(\sum_{\ell=0}^{k-1} \frac{W^\ell}{\ell!} + O(1) \cdot \left(\frac{|W|^k}{k!}\right)\cdot \max(1,\exp(\Re(W)))\right).
\end{align*}

Note that the expectation of the $\ell^{th}$ powers of $W$ are fooled by the $k$-wise independence of the $Y$'s for $\ell < k$. Therefore the difference in the expectations between the product of $Y$'s and the product of $X$'s is at most
\begin{equation}\label{taylorErrorEqn}
O(1) \cdot \prod_{i=1}^n \left(\mu_i\exp(\E[W_i])\right)\E\left[\left(\frac{|W|^k}{k!}\right) \cdot \max(1,\exp(\Re(W)))\right]
\end{equation}
Now, by Lemma \ref{lm:logexpcomplex} (4),
\begin{multline*}
\left|\prod_{i=1}^n \mu_i\exp(\E[W_i])\right| = \abs{\prod_{i=1}^n \mu_i \cdot \exp\bkets{\E\sbkets{\sum_i W_i}}} \leq \\
\abs{\prod_{i=1}^n \mu_i \cdot \E\sbkets{\exp\bkets{\sum_i W_i}}} = \E\left[ \left|\prod_{i=1}^n \mu_i \exp(W_i)\right|\right]  =  \E\left[ \left|\prod_{i=1}^nY_i\right|\right] \leq 1.
\end{multline*}
Further,
$$\left|\prod_{i=1}^n \mu_i\exp(\E[W_i])\right| \cdot \exp(\Re(W)) =  \left|\prod_{i=1}^n \mu_i\exp(\E[W_i]) \cdot \exp(W) \right| =   \left|\prod_{i=1}^nY_i\right| \leq 1.$$
Therefore, by Lemma \ref{lm:mombound}, the expression in \eqref{taylorErrorEqn} is at most
$$
O(1) \E\left[\frac{|W|^k}{k!}\right] \leq 2^{O(k)} \cdot \left(\frac{\tilde \sigma \sqrt{k} + B k}{k} \right)^k = 2^{O(k)} \cdot (\tilde\sigma/k^{1/2}+B)^k.
$$
\end{proof}

Next, we relax the conditions to handle the case where we only require the means of the $X_j$'s be far from zero.
\begin{lem}\label{largeLem}
Let $X_i$ and $Y_i$ be as in Equation \eqref{mainEqn}. Let $\mu_i=\E[X_i]$. If $|\mu_i| \geq (\sigma/\sqrt{k})^{1/3}$ for all $i$, then Equation \eqref{mainEqn} holds.
\end{lem}
\begin{proof}
We assume throughout that $\sigma/\sqrt{k}$ is less than a sufficiently small constant; otherwise, there is nothing to prove. Further, note that there can be at most $k$ different indices $j \in [n]$ where $\sigma_j \geq \sigma/\sqrt{k}$. As even after conditioning on the values of the corresponding $Y$'s, the remaining $Y_j$'s are $(C-1)k$-independent, it suffices to prove the lemma when $\sigma_j\leq \sigma/\sqrt{k}$ for all $j$.

To apply Lemma \ref{logLem}, we consider a truncation of our random variables: define
$$
\tilde Y_i = \begin{cases} Y_i &\textrm{if } |Y_i-\mu_i|\leq (\sigma/\sqrt{k})^{2/3} \\ \mu_i &\textrm{else} \end{cases}
$$
We claim that the variables $\tilde Y_i$ satisfy the conditions of Lemma \ref{logLem}. Let $\tilde \mu_i = \E[\tilde Y_i]$. Note that by Chebyshev bound, $\pr(\tilde Y_i\neq Y_i)\leq \sigma_i^2(\sigma/\sqrt{k})^{-4/3}\leq(\sigma/\sqrt{k})^{2/3}$. Therefore, $|\mu_i-\tilde \mu_i|\leq (\sigma/\sqrt{k})^{2/3}$, so that $|\tilde \mu_i| \geq (1/2)|\mu_i|$. Furthermore, letting $\tilde Y_i = \tilde \mu_i (1+Z_i)$, we have that
\begin{equation}\label{eq:largemu1}
\E[Z_i]=0,\;\; \|Z_i\|_\infty \leq 2 (\sigma/\sqrt{k})^{1/3},\;\;\var(Z_i)\leq  4\sigma_i^2 (\sigma/\sqrt{k})^{-2/3},\;\; \sum_i \var(Z_i) \leq 4 \sigma_i^2(\sigma/\sqrt{k})^{-2/3}.
\end{equation}

Finally, note that
$$\prod_{i=1}^n Y_i = \prod_{i=1}^n (Y_i - \tilde Y_i + \tilde Y_i) = \sum_{S\subseteq [n]}\prod_{i\in S}(Y_i-\tilde Y_i)\prod_{i\not\in S}\tilde Y_i.
$$
We truncate the above expansion to only include terms corresponding to sets $S$ with $|S| < m$ for $m = O(k)$ to be chosen later.
Let
$$P_m(Y_1,\ldots,Y_n) = \sum_{S\subseteq [n], |S| < m} \prod_{i \in S} (Y_i - \tilde Y_i) \prod_{i \notin S} (\tilde Y_i),$$
and let $N$ equal the number of $i$ so that $Y_i \neq \tilde Y_i$. We claim that
$$\abs{\prod_{j=1}^n Y_j - P_m(Y_1,\ldots,Y_n)} \leq 2^m \binom{N}{m}.$$
The above clearly holds when $N < m$, since in this case for any $S$ of size at least $m$ we have $\prod_{i\in S}(Y_i-\tilde Y_i)=0$. On the other hand for $N\geq m$ we note that there are at most $\sum_{\ell=0}^{m-1} \binom{N}{\ell} \leq 2^m\binom{N}{m}$ subsets $S$ for which this product is non-zero. Hence, $|P_m(Y_1,\ldots,Y_n)| < 2^m \binom{N}{m}$.

We now argue that $Ck$-wise independence fools the individual terms of $P_m$ when $m = O(k)$. This is because, the $Y_j$ for $j \in S$ are independent and conditioned on their values, the remaining $\tilde Y_j$ for $j \notin S$ are still $C'k$-wise independent for some sufficiently large constant $C'$. Therefore, applying Lemma \ref{logLem} with parameters as given by Equation \eqref{eq:largemu1}, $Ck$-wise independence fools $P_m(Y_1,\ldots,Y_n)$ up to error
$$
\sum_{S\subseteq [n], |S|< m}\prod_{i\in S}\left|\E[Y_i-\tilde Y_i]\right| \cdot 2^{O(k)} \bkets{\frac{\tilde \sigma^2}{\sqrt{k}}\bkets{\frac{\sigma}{\sqrt{k}}}^{-2/3}+\bkets{\frac{\sigma}{\sqrt{k}}}^{1/3}}^{3k},
$$
where
$$\bkets{\frac{\tilde \sigma^2}{\sqrt{k}}\bkets{\frac{\sigma}{\sqrt{k}}}^{-2/3}+\bkets{\frac{\sigma}{\sqrt{k}}}^{1/3}}^{3k} = O(\sigma/\sqrt{k})^k.$$
Therefore, $P_m$ is fooled to error
$$
\sum_{\ell=0}^{m-1}\E\left[\binom{N}{\ell} \right]2^{O(k)} \cdot (\sigma/\sqrt{k})^k.
$$

Note that the expectation above is the same as what it would be if the $Y_i$'s were fully independent, in which case it is at most
\begin{multline*}
\E[2^N] =\prod_{i=1}^n (1+\pr(Y_i\neq \tilde Y_i)) \leq \exp\bkets{\sum_{i=1}^n \pr(Y_i \neq \tilde Y_i)} = \\
\exp\left(O\left(\sum_{i=1}^n \sigma_i^2(\sigma/\sqrt{k})^{-4/3}\right)\right)  =\exp(O(\sigma^{2/3}k^{2/3}))=\exp(O(k)).
\end{multline*}
Therefore, $Ck$-wise independence fools $P_m$ to error $2^{O(k)} \cdot (\sigma/\sqrt{k})^k$.

On the other hand, the expectation of $\binom{N}{m}$ is
\begin{align*}
\sum_{S\subseteq [n],|S|=m} \prod_{i\in S} \pr(Y_i\neq \tilde Y_i) & \leq \frac{\left(\sum_{i=1}^n  \pr(Y_i\neq \tilde Y_i)\right)^m}{m!}\\
& \leq \frac{\left(\sum_{i=1}^n  \sigma_i^2(\sigma/\sqrt{k})^{-4/3}\right)^m}{m!}\\
& \leq O\left((\sigma^2/m) (\sigma^2/k)^{-2/3} \right)^m.
\end{align*}
Taking $m=3k/2$ yields a final error of $\exp(O(k)) \cdot (\sigma/\sqrt{k})^k$. This completes our proof.
\end{proof}

Finally, we can extend our proof to cover the general case.
\begin{proof}[Proof of Lemma \ref{lm:maintech}]
Note that it suffices to prove that Equation \eqref{mainEqn} holds. As before, it suffices to assume that $\sigma/\sqrt{k} \ll 1$ and that $\sigma_i \leq \sigma/\sqrt{k}$ for all $i$.

Let $m$ be the number of $i$ so that $|\E[Y_i]|\leq (\sigma/\sqrt{k})^{1/3}.$ Assume that the $Y$'s with \emph{small} expectation are $Y_1,\ldots,Y_m$. We break into cases based upon the size of $m$.

On the one hand if $m\leq 6k$, we note that for $C$ sufficiently large, the values of $Y_1,\ldots,Y_m$ are independent of each other, and even after conditioning on them, the remaining $Y_i$'s are still $C'k$-wise independent. Thus, applying Lemma \ref{largeLem} to the expectation of the product of the remaining $Y_i$ we find that the difference between the expectation of the product of $X$'s and product of $Y$'s is as desired.

For $m\geq 6k$ we note that
$$
\left|\E\left[\prod_{i=1}^n X_i \right] \right| = \prod_{i=1}^n\left|\E[Y_i] \right| \leq (\sigma/\sqrt{k})^{m/3}.
$$
Therefore, it suffices to show that
$$
\left|\E\left[\prod_{i=1}^n Y_i \right] \right| = O(\sigma/\sqrt{k})^k.
$$
Notice that so long as at least $3k$ of $Y_1,\ldots,Y_m$ have absolute value less than $2(\sigma/\sqrt{k})^{1/3}$, then
$$
\left|\prod_{i=1}^n Y_i \right| =O(\sigma/\sqrt{k})^k.
$$
Therefore, it suffices to show that this occurs except with probability at most $O(\sigma/\sqrt{k})^k$. Let $N$ be the number of $1\leq i \leq m$ so that $|Y_i|\geq 2(\sigma/\sqrt{k})^{1/3}.$ Note that
$$
\E[N] = \sum_{i=1}^m \pr(|Y_i|\geq 2(\sigma/\sqrt{k})^{1/3}) \leq \sum_{i=1}^m \sigma_i^2(\sigma/\sqrt{k})^{-2/3} \leq \sigma^2 (\sigma^2/k)^{-1/3}.
$$
On the other hand, we have that
\begin{align*}
\pr(N\geq 3k) & \leq \E\left[ \binom{N}{3k} \right]\\
& = \sum_{S\subseteq [m],|S|=3k} \prod_{i\in S} \pr(|Y_i|\geq 2(\sigma/\sqrt{k})^{1/3})\\
& \leq \frac{\left(\sum_{i=1}^m \pr(|Y_i|\geq 2(\sigma/\sqrt{k})^{1/3}) \right)^{3k}}{(3k)!}\\
& = \frac{\E[N]^{3k}}{(3k)!}\\
& \leq O((\sigma^2/k)^{2/3})^{3k}\\
& \leq O(\sigma/\sqrt{k})^k.
\end{align*}
This completes the proof.
\end{proof}

%% file: largeNorm.tex
% !TEX root = main.tex
\section{A Generator for high-variance \fshapes}
\label{sec:largenorm}

In this section, we construct a generator that fools Fourier shapes with high variance.
\begin{thm}\label{th:lnorm}
There exists a constant $C > 0$, such that  for all $\delta > 0$,
there exists an explicit generator $\calG_\ell:\zo^{r_\ell} \to [m]^n$
with seed-length $r_\ell = O(\log(mn/\delta)\log \log(1/\delta))$ such
that for all \fshapes\, $f:[m]^n \to \C_1$ with  $\tvar(f) \geq C
\log^5(1/\delta)$, we have
\[\abs{\E_{z \sim \zo^{r_\ell}}[f(\calG_\ell(z))] - \E_{X \in_u
    [m]^n}[f(X)]} < \delta.\]
\end{thm}

We start with the simple but crucial observation that Fourier shapes
with large variance have small expectation.

\begin{lem}
For any \fshape\, $f: [m]^n \rgta \C_1$, we have
\begin{align}
\label{eq:var-bound}
\abs{\E_{X \in_u [m]^n}[f(X)] }\leq \exp(-\tvar(f)/2).
\end{align}
\end{lem}
\begin{proof}
Let $f(x) = \prod_j f_j(x_j)$. Since $f_j(x) \in \C_1$, we have
$|f_j(x)| \leq 1$. Let  $\mu_j = \E_{X_j \in [m]}[f_j(X_j)]$. For $X \in_u [m]^n$,
\begin{align*}
\sigma_j^2  = \E[|f_j(X_j) - \mu_j|^2] = \E[|f_j(X_j)|^2] - |\mu_j|^2 \leq 1 - |\mu_j|^2.
\end{align*}
Hence
\begin{align*}
\abs{\E_X[f(X)]} = \prod_{j=1}^n \abs{\mu_j} &\leq \prod_{j=1}^n( 1 - \sigma^2_j)^{1/2}\\
&\leq \exp(- \sum_{j=1}^n \sigma^2_j/2) \leq \exp(-\tvar(f)/2).
\end{align*}
\end{proof}

We build the generator in two steps. We first build a
generator with seed-length $O(\log n)$ which achieves constant error
for all $f$ with $\tvar(f) \geq 1$. In the
second step, we reduce the error down to $\delta$. This
construction is inspired by a construction of Naor and
Naor \cite{NaorN93} of small-bias spaces.% the difference being that we need generators for space bounded algorithms for various amplification steps, as opposed to expander random walks.

\subsection{A generator with constant error}

Our goal in this subsection is get a generator with constant error for
\fshapes\ where $\tvar(f) = \Omega(1)$. We start by showing that when $\tvar(f) = \Theta(1)$ (instead of
just $\Omega(1)$),  $O(1)$-wise independence is enough to fool $f$.

\begin{lem}\label{lm:vconstnorm}
For all constants $0 < c_1 < c_2 $, there exist $p \in \Z_+$ and
$0 < c' < 1$ such that the following holds. For any $(m,n)$-\fshape, $f$ with $\tvar(f) \in
[c_1, c_2]$, and $Z \sim [m]^n$ $2p$-wise independent,
\[ \abs{\E_{Z}[f(Z)]} < c'.\]
\end{lem}
\begin{proof}
Let $f = \prod_j f_j$, $X \in_u [m]^n$. Now, by Lemma \ref{lm:maintech} applied to $Y_j = f_j(Z_j)$, we have,
$$\abs{\E[f(Z)] - \E[f(X)]} \leq \exp(O(p)) (\tvar(f)/\sqrt{p})^{\Omega(p)} = \exp(O(p))(c_2/\sqrt{p})^{\Omega(p)}. $$
Note that by taking $p$ to be a sufficiently large constant compared to $c_2$, we can make the last bound arbitrary small.

On the other hand, by Equation \eqref{eq:var-bound},
\[ \abs{\E[f(X)]} \leq \exp(-\tvar(f)/2) \leq \exp(-c_1/2).\]
Therefore,
\[ \abs{\E[f(Z)]} \leq \exp(-c_1/2) + \exp(O(p)) (c_2/\sqrt{p})^{\Omega(p)} < c'\]
for $p$ sufficiently large constant and some constant $0 < c' < 1$.
\ignore{
The claim now
We first bound  away from $1$ when $X \in_u \dpm^n$ using Equation \eqref{eq:var-bound}.
\begin{equation}
\label{eq:lnorm1}
\abs{\E_{X}f(X)} = \abs{\prod_{i=1}\mu_i} \leq \exp(-c_1/2).
\end{equation}

We now use Lemma \ref{lm:maintech} to conclude that
\[ \abs{\E_Yf(Y) - \prod_{i=1}\mu_i} \leq \exp(O(k)) (c_2/\sqrt{k})^k.\]
By taking $k$ to be sufficiently large constant compared to $c_2$, we
can make the RHS arbitrarily small. Hence we have
\[ \abs{\E_Yf(Y)} \leq  \abs{\prod_{i=1}\mu_i} + c_3 < c'.\]}
\end{proof}

We reduce the general case of $\tvar(f) \in [1,n]$ to the case above
where $\tvar(f) = \Theta(1)$ by using the Valiant-Vazirani technique
of sub-sampling. For $B \subseteq [n]$ let $\tvar(f_B) = \sum_{i \in
  B}\sigma_i^2$. If we  sample a random subset $B \subseteq [n]$ with $|B| \approx
n/\tvar(f)$ in a pairwise independent manner,  we will get
$\tvar(f_B) = \Theta(1)$ with $\Omega(1)$ probability.
Since we do not know $\tvar(f)$, we sample $\log(n)$ subsets whose
cardinalities are geometrically increasing; one of them
is likely to satisfy  the desired bound.

We set up some notation that will be used in the remainder of this section.
\begin{itemize}
\item Assume $n$ is a power of $2$, and set $T = \log_2(n) -1$.
Let $\Pi \subseteq \mathbb{S}_n$ be a family of pairwise independent permutations so
that $\pi \in_u \Pi$ can be sampled efficiently with $O(\log n)$
random bits. For $0 \leq j \leq T$, let $B_j = \{\pi(i): i
\in \{2^j,\ldots,2^{j+1}-1\}\}$ be the $2^j$ co-ordinates that land in
the $j^{th}$ bucket.
\item For $v \in \R^n$, let $v^j = v_{B_j}$ denote the projection of $v$
onto coordinates in bucket $j$. Similarly, for $x \in [m]^n$, let $x^j$ denote the
projection of $x$ to the co-ordinates in $B_j$.
\item Fix an $(m,n)$-\fshape\ $f:[m]^n \to \C_1$ with $f(x) =\prod_i
f_i(x_i)$. Define $f^j:[m]^{B_j} \rgta \C_1$ as $f^j(x^j) = \prod_{i \in B_j}f_i(x_i)$.
\end{itemize}

\begin{lem}
\label{lem:pw-sampling}
Let $v \in \R^n$ with $\nmt{v}^2 \in [1,n]$, $\|v\|_\infty\leq 1$ and
$t \in [\log_2 n]$ be such that $n/2^{t+1} \leq \nmt{v}^2 \leq
n/2^t$. Then,
\[\Pr_{\pi \in_u \Pi}\left[\nmt{v^t}^2 \in[1/6,4/3]\right] \geq 7/16.\]
\end{lem}

The proof of this lemma is standard and is deferred to Appendix
\ref{app:largenorm}.

This naturally suggests using an $O(1)$-wise independent distribution within
each bucket. But using independent strings across the $\log(n)$ buckets would
require a seed of length $O(\log(mn) \cdot (\log n))$. We analyze our generator
assuming independence across distinct buckets, but then
recycle the seeds using $\prg$s for space bounded computation to keep the
seed-length down to $O(\log(mn))$ (rather than $O(\log^2(n))$).

We now prove the main claim of this subsection.

\begin{lem}\label{lm:lnormconst}
There exists an explicit generator $\calG_1:\zo^r \to [m]^n$ with $r =
O(\log (mn))$ such that for all \fshapes\, $f:[m]^n \rgta \C_1$ with $\tvar(f) \geq 1$, we have
\[ \abs{\E_{z \sim \zo^r}[f(\calG_1(z))]} \leq c.\]
for some constant $0 < c < 1$.
\end{lem}
\begin{proof}
Let $\pi \in_u \Pi$ and let $Z^j \sim [m]^{2^j}$ be an independent $p$-wise independent string for a parameter $p = O(1)$ to be chosen later. Define
\[ \calG_1'(\pi,Z^0,\ldots,Z^T) = Y, \;\;\; \text{ where} \ \ Y_{B_j}
= Z^j \  \ \text{for}  \ j \in \{0,\ldots, T\}.\]
In other words, the generator applies the string $Z^j$ to the
coordinates in bucket $B_j$.

Observe that $f(Y) = \prod_{j=0}^{\log(n) -1} f^j(Z^j)$. Since the
$Z^j$'s are independent of each other
\[
\abs{\E[f(Y)]} = \abs{\prod_{j=0}^{\log(n) -1} \E[f^j(Z^j)]} \leq \abs{\E[f^t(Z^t)]},
\]
for any $t \leq T$. Applying Lemma \ref{lm:vconstnorm} to $v= (\sigma_1(f_1),\ldots,\sigma_n(f_j))$,
we get that for some $t \leq T$, $\tvar(f^t) = \nmt{v^t}^2 \in [1/6,4/3]$ with probability at least $7/16$. Conditioned on this
event, Lemma \ref{lm:vconstnorm} implies that for $p$ a sufficiently large constant, there exists a constant $c' < 1$ so that
$\abs{\E[f^t(Z^t)]} < c'$. Therefore, overall we get
\[ \abs{\E[f(Y)]} \leq \abs{\E[f^t(Z^t)]} \leq \frac{9}{16} + \frac{7c'}{16} = c'' < 1.\]

We next improve the seed-length of $\calG_1'$ using the $\prg$ for
ROBPs of Theorem \ref{th:nz}. To this end, note that by Lemma \ref{lem:bounded} we can
assume that every $\log(f_i(x_i))$, and hence every $\log(f^j(x^j))$, has bit
precision at most $O(\log n)$ bits (since our goal is to get error
$\delta =O(1)$). Further, each $Z^j$ can be generated
efficiently with $O(\log(mn))$ random bits.

Thus, for a fixed permutation $\pi$, the computation of $f(\calG'(\pi,Z^1,\ldots,Z^T))$ can
be done by a $(S,D,T)$-ROBP where $S,T$ are $O(\log n)$ and $D =
O(\log (mn))$: for $j \in \{1,\ldots,T\}$, the ROBP computes
$f^j(Z^j)$ and multiplies it to the product computed so far, which can
be done using $O(\log n)$ bits of space. Let $\calG^{NZ}:\zo^r \to
\left(\zo^D\right)^T$ be the generator in Theorem \ref{th:nz} fooling
$(S,D,T)$-ROBPs as above with error $\delta < (1 - c'')/2$. $\calG^{NZ}$ has seedlength $O(\log(mn))$. Let
\[\calG_1(\pi,z) = \calG_1'(\pi,\calG^{NZ}(z)).\]

It follows  that $\abs{\E[f(\calG_1(\pi,z))]}< c$ for some constant $c
< 1$. Finally, the seed-length of $\calG_1$ is
$O(\log (mn))$ as $\pi$ can be sampled with $O(\log n)$ random bits
and the seed-length of $\calG^{NZ}$ is $O(\log (m n))$. The lemma is now
proved.
\end{proof}

\subsection{Reducing the error}

We now amplify the error to prove Theorem \ref{th:lnorm}. The starting point for the construction is the
observation that for $X \in_u [m]^n$, $\abs{\E[f(X)]} \leq \exp(-\tvar(f/2)) \leq \delta$ once $\tvar(f) \gg
\log(1/\delta)$. Therefore, it suffices to design a generator so that
$\E[f] \ll \delta$, when $\tvar(f)$ is sufficiently large.

Our generator will partition $[n]$ into $m = O((\log(1/\delta))^5)$
buckets $B_1,\ldots,B_m$, using a family of hash functions with the
following \emph{spreading} property:
\begin{definition}
 A family of hash functions $\hh = \{h:[n] \to [m]\}$ is said to be
 $(B,\ell,\delta)$-spreading if for all $v \in [0,1]^n$ with
 $\nmt{v}^2 \geq B$,
 \[ \Pr_{h \in_u \hh}[|\{j \in [m]: \nmt{v_{h^{-1}(j)}}^2 \geq B/2m\}|
   \geq \ell] \geq 1 - \delta.\]
 \end{definition}

Using the notation from the last subsection, we write $f(x) =
\prod_{j=1}^mf^j(x^j)$ where $f^j(x^j) =\prod_{i \in B_j}f_i(x_i)$. If $\tvar(f)$ is
sufficiently large, then the spreading property guarantees that for at least $\Omega(\log(1/\delta))$ of the
buckets $B_j$, $\tvar(f^j) \geq 1$. If we now generate $X \in [m]^n$ by  setting $X_{B_j}$ to be
an independent instantiation of the generator $\calG_1$ from Lemma
\ref{lm:vconstnorm}, then  we get $\E[f(X)] \ll \delta$. As in
the proof of Lemma \ref{lm:lnormconst}, we keep the seed-length down
to $\tilde{O}(\log(n/\delta))$ by recycling the seeds for the buckets
using a $\prg$ for small-space machines.

We start by showing that the desired hash functions can be generated
from a small-bias family of hash functions. We show that it satisfies the conditions
of the lemma by standard moment bounds. The proof is in Appendix
\ref{app:largenorm}

\begin{lem}\label{SpreadingHashLem}
For all constants $C_1$, there exist constants $C_2,C_3$ such that
following holds. For all $\delta \geq 0$, there exists an explicit
hash family $\hh = \{h:[n] \to [T]\}$, where $T = C_2
\log^5(1/\delta))$ which is $(C_3 \log^5(1/\delta), C_1
\log(1/\delta), \delta)$-spreading and $h \in_u \hh$ can be sampled
efficiently with $O(\log(n/\delta))$ bits.
\end{lem}

We are now ready to prove Theorem \ref{th:lnorm}.
\begin{proof}[Proof of Theorem \ref{th:lnorm}]
Let $\ell = C \log(1/\delta)$ for some constant to be chosen later and
let $\hh = \{h:[n] \to [T]\}$ be a $(B,\ell,\delta)$-spreading family
as in Lemma \ref{SpreadingHashLem} above for $B =
\Theta(\log^5(1/\delta))$ and $T = \Theta(\log^5(1/\delta))$. Let
$\calG_1:\zo^{r_1} \to [m]^n$ be the generator in Lemma
\ref{lm:vconstnorm}. Define a new generator $\calG_\ell':\hh \times
(\zo^{r'})^T \to [m]^n$ as:
\[ \calG_\ell'(h,z^1,\ldots,z^T) = X,\;\;\;\text{ where $X_{h^{-1}(j)} =
  \calG_1(z^j)$ for $j \in [T]$}.\]

Let $f: [m]^n \rgta \C_1$ with $\tvar(f) \geq \max(2T,B)$. For $h \in
\hh$, let $I = \{j: \tvar{f^j} \geq 1\}$. For any fixed $h \in \hh$,
as the $z^j$'s are independent of each other,
\[ \abs{\E[f(X)]} = \prod_{j=1}^m \abs{\E[f^j(\calG_1(z^j))]} \leq \prod_{j \in I} \abs{\E[f^j(\calG_1(z^j))]} \leq c^{|I|},\]
where $c < 1$ is the constant from Lemma \ref{lm:vconstnorm}. By the
spreading property of $\hh$, with probability at least $1-\delta$,
$|I| \geq C \log(1/\delta)$. Therefore, for $C$ sufficiently large,
\[\abs{\E[f(X)]} \leq \delta + c^{C \log(1/\delta)} < 2 \delta.\]

As in Lemma \ref{lm:lnormconst}, we recycle the seeds for the various
buckets using the PRGs for ROBPs. By Lemma \ref{lem:bounded}, we may assume that  $f^j$ has bit
precision at most $O(\log (n/\delta))$ bits.  Further note that
\[ f(\calG_\ell'(h,z^1,\ldots,z^m)) = \prod_{j=1}^mf^j(\calG_1(z^j)). \]
For a fixed hash function $h \in \hh$, this
can be computed by a $(S,D,T)$-ROBP where $S = O(\log(n/\delta))$ and
$D = O(\log (mn))$, corresponding to the various possible seeds for
$\calG_1$. Let $\calG^{INW}:\zo^r \to
\left(\zo^D\right)^T$ be a generator fooling $(S,D,T)$-ROBPs as in
Theorem \ref{th:inwprg} with error $\delta$ and define
\[ \calG_\ell(h,z) = \calG_\ell'(h, \calG^{INW}(z)).\]
The seed-length is dominated by the seed-length of $\calG^{INW}$, which is
\[ O(\log(mn/\delta)\log T) = O(\log(mn/\delta)\log\log(1/\delta)).\]
It follows that $\abs{\E[f(\calG_\ell(h,z))]} < 3 \delta$, whereas for a truly random $Y \in_u [m]^n$,
\[ \abs{\E[f(Y)]} \leq \exp(-\tvar(f)/2) < \delta.\]
The theorem now follows.
\end{proof}

%%% Local Variables:
%%% mode: latex
%%% TeX-master: "main"
%%% End:

%% file: alphabetredux.tex
% !TEX root = main.tex
\section{Alphabet reduction for \fshapes}

In this section, we describe our alphabet-reduction procedure, which
reduces the general problem of constructing an $\epsilon$-PRG for
$(m,n)$-\fshapes\ where $m$ could be much larger than $n$, to that of
constructing an $\epsilon/\log(m)$-PRG for $(n^4,n)$-\fshapes. This
reduction is composed of $O(\log \log m)$ steps where in each step we
reduce fooling $(m,n)$-\fshapes\ to fooling $(\sqrt{m},n)$-\fshapes. Each of these
steps in turn will cost $O(\log(m/\epsilon))$ random bits, so that the
overall cost is $O(\log(m/\epsilon) \cdot (\log \log m))$. Concretely, we show the following:

\begin{thm}\label{th:areduxmain}
Let $n, \delta > 0$ and suppose that for some $r' = r'(n,\delta')$,  for all $m' \leq n^4$ there exists an explicit
generator $\calG_{m'}:\zo^{r_1} \to [m']^n$ which $\delta'$-fools
$(m',n)$-\fshapes. For all $m$, there exists an explicit
generator $\calG_m:\zo^r \to [m]^n$ which $(\delta' + \delta)$-fools
$(m,n)$-\fshapes\, with seed-length $r = r' + O(\log(m/\delta)\log\log(m))$.
\end{thm}
\begin{proof}
We prove the claim by showing that for $m > n^4$, we can reduce
$(\delta+\delta')$-fooling $(m,n)$-\fshapes\, to that of
$\delta'$-fooling $(\sqrt{m},n)$-\fshapes\, with $O(\log(m/\delta))$
additional random bits. The theorem follows by  applying
the claim $\log\log(m)$ until the alphabet size drops below $n^4$ when
we can use $\calG_{m'}$. This costs a total of
$r' + O(\log(m/\delta)\log\log(m))$ random bits, and gives error
$\delta' + \log\log(m)\delta$. The claim follows by replacing $\delta$
with $\delta/\log\log(m)$.

Thus, suppose that $m > n^4$ and for $D = \floor{\sqrt{m}}$, we have a
generator $\calG_D:\zo^{r_D} \to [D]^n$ which $\delta'$-fools
$(D,n)$-\fshapes. The generator $\calG_m$ works as follows:
\begin{enumerate}
\item Generate a matrix $X \in [m]^{D \times n}$ where
\begin{itemize}
\item Each column of $X$ is from a pairwise independent distribution over $[m]^D$.
\item The different columns are $k$-wise independent for $k = C\log(1/\delta)/\log (m)$ for some sufficiently large constant $C$.
\end{itemize}
\item Generate $Y = (Y_1,\ldots,Y_n) = \calG_D(z) \in [D]^n$ for $z \in_u \zo^{r_D}$.
\item $\calG_m$ outputs $Z = (Z_1,\ldots,Z_n) \in [m]^n$ where $Z_j = X[Y_j, j]$ for $j \in [n]$.
\end{enumerate}
Each column of $X$ can be generated  using a seed of length $2\log
m$. By using seeds for various columns that are $k$-wise independent,
generating $X$ requires seedlength $O(k\log m) = O(\log(1/\delta))$
(as $m > n^2$), while  the number of bits needed to generate $Z$ is
$r_D + O(\log(1/\delta))$.

Fix an $(m,n)$-\fshape\, $f:[m]^n \to \C_1$, $f(z) = \prod_j f_j(z_j)$. For
$x \in [m]^{D \times n}$, define a $(D,n)$-\fshape\, $f^x:[D]^n \to
\C_1$ by:
\[f^x(y_1,\ldots,y_n) = \prod_{j=1}^n f_j(x[y_j,j]).\]
Note that $f(Z) = f^X(Y)$.

Let $X', Y'$ be random variables
distributed uniformly over $[m]^{D \times n}$ and $[D]^n$
respectively. Let $Z'_j= X'[Y'_j,j]$ for $j \in [n]$, so that $Z'$
is uniform over $[m]^n$ and $f(Z') = f^{X'}(Y')$. Our goal is to show that $f(Z')$ and $f(Z)$
are close in expectation. We do this by replacing $X'$ and $Y'$ by $X$
and $Y$ respectively.

That we can replace $Y'$ with $Y$ follows from the pseudorandomness of $\calG_D$.
For any fixed $x \in [m]^n$, as $\calG_D$ fools $(D,n)$-\fshapes,
\begin{equation}\label{eq:aredux1}
\abs{\E_{Y = \calG_D(z)}[f^x(Y)] - \E_{Y' \in_u [D]^n}[f^x(Y')]} \leq \delta'.
\end{equation}
We now show that for truly random $Y'$, one can replace $X$ by
$X'$. Note that
\begin{equation}\label{eq:aredux2}
\E_{Y' \in_u [D]^n}[f^x(Y')] = \prod_{j=1}^n \bkets{\frac{1}{D} \cdot \bkets{\sum_{\ell=1}^D f_j(x[\ell, j])}} \equiv B_f(x).
\end{equation}
where we define the {\em bias-function} $B_f:[m]^{D \times n} \to
\C_1$ as above. We claim that $X$ fools $B_f$:
\begin{align}
\label{eq:aredux3}
\abs{\E[B_f(X)] - \E[B_f(X')]} \leq \delta.
\end{align}
For $j \in [n]$, let
\begin{align*}
A_j  = \frac{1}{D} \bkets{\sum_{\ell=1}^D f_j(X[\ell,j])},\
A'_j  = \frac{1}{D} \bkets{\sum_{\ell=1}^D f_j(X'[\ell,j])}
\end{align*}
so that
\begin{align*}
B_f(X) = \prod_{j=1}^nA_j, \ B_f(X') = \prod_{j=1}^nA_j'.
\end{align*}
Since $f_j(X[\ell,j]) \in \C_1$ for $\ell \in [D]$, it follows that
$A_j,A_j' \in \C_1$. Since the $f_j(X[\ell,j])$s are pairwise
independent variables,
\[ \E[A_j] = \E[A_j'], \ \var[A_j] = \var[A_j'].\]
Note that
\begin{align}
\label{eq:exp_b}
\E[B_f(X')] = \E[\prod_{i=1}^nA'_j] =
\prod_{j=1}^n\E[A_j'] = \prod_{j=1}^n\E[A_j], \ \E[B_f(X)]  = \E[A_1 \cdots A_n].
\end{align}
The random variables $A_1,\ldots,A_n$ are $k$-wise
independent. Further, we have
\begin{align*}
\var(A_j) = \frac{1}{D^2} \sum_{\ell=1}^D \var(f_j(X[\ell, j])) =
\frac{\sigma^2(f_j)}{D} \leq \frac{1}{D}.
\end{align*}
Therefore, by Lemma \ref{lm:maintech},
\begin{align}
\label{eq:prod_a}
\abs{\E[A_1 \cdots A_n] - \prod_{j=1}^n \E[A_j]}  \leq \left(\frac{n}{D}\right)^{\Omega(k)}
\leq m^{-\Omega(k)} \leq \delta
\end{align}
where the second to last inequality follows becase $n \leq m^{1/4}$
and $D \geq \sqrt{m}/2$, and the last holds for $k
=C\log(1/\delta)/\log(m)$ for a sufficiently big constant $C$.
Equation \ref{eq:aredux3} now follows from Equations \eqref{eq:prod_a}
and \eqref{eq:exp_b}.

Finally,
\begin{align*}
\abs{\E[f(Z)] - \E[f(Z')]} &= \abs{\E[f^X(Y)] - \E[f^{X'}(Y')]}& \\
&=  \abs{\E[f^X(Y')] - \E[f^{X'}(Y')]} + \delta'\ & \text{Equation \eqref{eq:aredux1}}\\
&= \abs{\E[B_f(X)] - \E[B_f(X')]} + \delta'\ & \text{Equation \eqref{eq:aredux2}}\\
& \leq \delta + \delta'. \ & \text{Equation \eqref{eq:aredux3}}
\end{align*}
Hence the theorem is proved.
\end{proof}

%% file: dimredux.tex
% !TEX root = main.tex
\section{Dimension reduction for low-variance \fshapes}
\label{sec:dimredux}

We next describe our \emph{dimension reduction} step for low-variance
\fshapes. We start with an $(m,n)$-\fshape\ where $m \leq n^4$ and $\tvar(f) \leq \log(n/\delta)^c$. We
show how one can reduce the dimension to $t = \sqrt{n}$, at a price of a
blowup in the alphabet size $m'$ which now becomes $(n/\delta)^c$ for some (large) constant $c$.

\begin{thm}\label{th:dimreduxlowvar}
Let $\delta > 0$, $n > 0$ and $t = \lceil \sqrt{n} \rceil$. There is a
constant $c$ and $m' \leq (n/\delta)^c$ such that the following holds:
if there exists an explicit $\prg$ $\calG': \zo^{r'} \to [m']^t$ with seed-length $r' = r'(n,\delta')$  which $\delta'$-fools
  $(m',t)$-\fshapes, then there exists an
  explicit generator $\calG:\zo^r \to [m]^n$ with seed-length $r = r' +
  O(\log(n/\delta))$ which $(\delta + \delta')$-fools $(m,n)$-\fshapes\, $f$ with
  $m \leq n^4$ and $\tvar(f) \leq n^{1/9}$.
\end{thm}

We first set up some notation. Assume that we have fixed a hash function $h: [n] \to
[t]$. For $x \in [m]^n$ and $j \in [t]$, let $x^j$ denote the
projection of $x$ onto co-ordinates in $h^{-1}(j)$. For an
$(m,n)$-\fshape\,$f:[m]^n \to \C_1$ with $f = \prod_{i=1}^n f_i$, let
\begin{align*}
f^j(x^j) & = \prod_{i:h(i) =j}f_i(x_i)\\
\text{so that }\ f(x) & = \prod_{j=1}^t f^j(x^j).
\end{align*}

We start by constructing an easy to analyze generator $G_1$ which hashes co-ordinates into
buckets using $k$-wise independence and then uses independent $k$-wise
independent strings within  a bucket. Let
\begin{align}
\label{eq:def-c}
k = C \frac{\log(n/\delta)}{\log(n)}
\end{align}
where $C$ will is a sufficiently large constant.
Let $\hh:\{[n] \to t\}$ be a $k$-wise independent
family of hash functions. Let $G_0:\zo^{r_0} \to [m]^n$ be a $k$-wise independent generator over
$[m]^n$. Define a new generator $G_1:\hh \times (\zo^{r_0})^t \to [m]^n$ as:
\begin{equation}
\label{eq:dimreduxg1}
G_1(h,z_1,\ldots,z_t) = Z, \text{where } Z^j = G_0(z_j)
\ \forall \ j \in [t].
\end{equation}

We argue that $G_1$ fools $(m,n)$-\fshapes\, with small total variance as in the theorem.
Our analysis proceeds as follows:
\begin{itemize}
\item With high probability over $h\in_u \hh$, each of the $f^j$'s has
low variance except for a few heavy co-ordinates (roughly $\tvar(f)/t$
after dropping $k/2$ heavy coordinates).
\item Within each bin we have $k$-wise independence, whereas the
  distributions across bins are independent. So even conditioned
  on the heavy co-ordinates in a bin, the remaining distribution in
  the bin is $k/2$-wise independent. Hence each $f^j$ is fooled by Lemma \ref{lm:maintech}.
\end{itemize}
However, the seed-length of $G_1$ is prohibitively large: since we use
independent seeds across the various buckets, the resulting seed-length
is $O(\sqrt{n}\log(n/\delta))$. The crucial observation is that
we can {\em recycle} the seeds for various buckets using a generator that fools $(m',t)$-\fshapes\ with
$m' = 2^{r_0} = \poly(n/\delta)$ and $t =O(\sqrt{n})$. Given such a generator
$\calG':\zo^{r'} \to [m']^t$ which $\delta$-fools $(m',t)$-\fshapes,
our final generator for small-variance \fshapes\ is $\calG_s:\hh \times
\zo^{r'} \to [m]^n$ is defined as
\begin{align}
\calG(h,w) = G_1(h, \calG'(w)).% \text{where} \ z \in_u \zo^{r'}.
\end{align}
It is worth mentioning that even though the original \fshape\ $f:[m]^n
\rgta \C_1$ has low total variance, the generator $\calG'$ needs to fool all
$(m',t)$-\fshapes, not just those with low variance.

\subsection{Analysis of the dimension-reduction step}

For $\alpha > 0$, to be chosen later, let $L = \{j
\in [n]: \sigma^2(f_j) \geq \alpha\}$ denote the $\alpha$-large indices and $S
= [n]\setminus L$ denote the small indices. We call a hash function $h \in
\hh$ \emph{$(\alpha,\beta)$-good} if the following two conditions hold
for every bin $h^{-1}(j)$ where $j \in [t]$:
\begin{enumerate}
\item The bin does not have too many large indices: $|h^{-1}(j) \cap L| \leq k/2$.
\item The small indices in the bin have small total variance:
\[ \sum_{\ell \notin L: h(\ell) = j} \sigma^2(f_\ell) \leq \beta.\]
\end{enumerate}

Using standard moment bounds for $k$-wise independent hash functions
one can show that $h \in_u \hh$ is $(\alpha,\beta)$-good
with probability at least $1-n^{-\Omega(k)}$ for $\alpha =
n^{-\Omega(1)}$ and $\beta = n^{-\Omega(1)}$.  We defer the proof of
the following Lemma to Appendix \ref{app:dimredux}.

\begin{lem}\label{lm:goodh}
Let $\tvar(f) \leq n^{1/9}$ and let $\hh = \{h:[n] \to [t]\}$ be a $k$-wise independent family of hash functions for $t=\Theta(\sqrt{n})$.
Then $h \in \hh$ is
$(n^{-1/3},n^{-1/36})$-good with probability $1-O(k)^{k/2} n^{-\Omega(k)}$.
\end{lem}

We next argue that if $h \in \hh$ is $(\alpha,\beta)$-good then, $k$-wise independence is sufficient to fool $f^j$ for each $j \in [t]$.
\begin{lem}\label{lm:foolgoodh}
Let $h \in \hh$ be $(\alpha,\beta)$-good, and let $j \in [t]$. For $Z' \sim [m]^{n}$ $k$-wise independent, and $Z'' \in_u [m]^{n}$,
\[\abs{\E[f^j(Z')] - \E[f^j(Z'')]} \leq \exp(O(k)) \cdot \beta^{\Omega(k)}.\]
\end{lem}
\begin{proof}
Fix $j \in [t]$. By relabelling coordinates, let us assume that
$h^{-1}(j) = \{1,\ldots,n_j\}$ and $L \cap h^{-1}(j) =
\{1,\ldots,r\}$, where $r \leq k/2$.
As $Z'$ is $k$-wise independent,
$(Z'_1,\ldots,Z'_r)$ is uniformly distributed over $[m]^r$. We couple
$Z'$ and $Z''$ by taking $Z'_i = Z''_i$ for $i \leq r$. Even after
conditioning on these values, $Z'_{r+1},\ldots,Z'_{n_j}$ are $k/2$-wise
independent.

Let $Y_\ell = f_\ell(Z'_\ell)$ for $\ell \in \{r+1,\ldots,n_j\}$. As $h$ is $(\alpha,\beta)$-good,
\[ \sum_{\ell=r+1}^{n/t} \sigma^2(Y_\ell) \leq \beta.\]
Therefore, by Lemma \ref{lm:maintech},
\begin{align}
\label{eq:k-wise}
\abs{\E\sbkets{\prod_{\ell =r+1}^{n_j}Y_\ell} - \prod_{\ell=r+1}^{n_j}
  \E[Y_\ell]} \leq \exp(O(k)) \cdot \beta^{\Omega(k)}.
\end{align}
But since $Z'_\ell =Z'_\ell$ for $\ell \leq r$, we have
\begin{align*}
\E[f^j(Z')] & = \prod_{\ell=1}^r\E[f^\ell(Z'_\ell)]\E[\prod_{\ell=r+1}^nY_\ell],\\
\E[f^j(Z'')] & =  \prod_{\ell=1}^r\E[f^\ell(Z'_\ell)]\prod_{\ell=r+1}^{n_j}\E[Y_\ell],\\
\abs{\E[f^j(Z')] - \E[f^j(Z'')] } &  =
\abs{\prod_{\ell=1}^r\E[f^\ell(Z'_\ell)]\E[\prod_{\ell=r+1}^nY_\ell] -
\prod_{\ell=1}^r\E[f^\ell(Z'_\ell)]\prod_{\ell=r+1}^{n_j}\E[Y_\ell]}&\\
& \leq  \abs{\E[\prod_{\ell=r+1}^{n_j}Y_\ell] -
  \prod_{\ell=r+1}^{n_j}\E[Y_\ell]} & \text{Since } |f^\ell(Z'_\ell)|  \leq 1\\
& \leq \exp(O(k)) \cdot \beta^{\Omega(k)}. & \text{Equation } \eqref{eq:k-wise}
\end{align*}
\end{proof}

We use these lemmas to prove Theorem \ref{th:dimreduxlowvar}.

\begin{proof}[Proof of Theorem \ref{th:dimreduxlowvar}]

%We can assume that $n \geq c \log(1/\delta)^c$; otherwise, the desired generator is implied by Theorem \ref{th:inwprg}. 
Let $f:[m]^n \rgta \C_1$ be a \fshape\ with $\tvar(f) \leq n^{1/9}$.
Let $G_1$ be the generator in Equation \eqref{eq:dimreduxg1} with
parameters as above. We condition on $h \in_u \hh$ being $(n^{-1/3}, n^{-1/36})$-good;
by Lemma \ref{lm:goodh} this only adds an additional $O(k)^{k/2}
n^{-\Omega(k)}$ to the error. We fix such a good hash function $h$.

Recall that $G_1(h,z_1,\ldots,z_t) =Z$ where $Z^j = G_0(z_j)$ for $j \in [t]$. Since the $z_j$s are independent,
so are the $Z^j$'s. Hence,
\begin{align*}
\E[f(G_1(h,z^1,\ldots,z^t))] &= \prod_{j=1}^t\E_h \sbkets{f^j(Z^j)}.
\end{align*}

By Lemma \ref{lm:foolgoodh}, for $(n^{-1/3}, n^{-1/36})$-good $h$, if $Y \in_u [m]^n$, then
\begin{align*}
& \abs{\prod_{j=1}^t\E \sbkets{f^j(Z^j)} - \prod_{j=1}^t\E
  \sbkets{f^j(Y^j)}} \\
& \leq \sum_{r=0}^{t-1} \abs{\prod_{j=1}^r\E
  \sbkets{f^j(Y^j)}\prod_{j=r+1}^t\E \sbkets{f^j(Z^j)} - \prod_{j=1}^{r+1}\E
  \sbkets{f^j(Y^j)}\prod_{j=r+2}^t\E \sbkets{f^j(Z^j)}}\\
& \leq \sum_{r=0}^{t-1}\abs{\E \sbkets{f^{r+1}(Z^{r+1})} - \sbkets{f^{r+1}(Y^{r+1})}}\\
& \leq \exp(O(k)) \cdot O(tn^{-k/36}).
\end{align*}
Combining the above equations we get that for $Y \in_u [m]^n$,
\begin{align}
\label{eq:g1-works}
\abs{\E[f(G_1(h,z_1,\ldots,z_t))] - \E[f(Y)]} \leq  O(k)^{k/2}
n^{-\Omega(k)} + \exp(O(k)) \cdot O(tn^{-k/36}) \leq \delta
\end{align}
where the last inequality holds by taking $C$ in Equation
\eqref{eq:def-c} to be a sufficiently large constant.

We next derandomize the choice of the $z^j$'s by using a PRG for
appropriate \fshapes. Let $r_0$ be the seed-length of the generator
$G_0$ obtained by setting $k = C \log(n/\delta)/(\log n)$ as above,
and let $c$ be such that $r_0 \leq c\log(n/\delta)$. Let
\[ m' = 2^{r_0} \leq \left(\frac{n}{\delta}\right)^c \]
and identify $[m']$ with $\zo^{r_0}$.
Given a hash function $h \in \hh$, let us define
$\bar{f}^j:[m'] \rgta \C_1$ for $j \in [t]$ and $\bar{f}: [m']^n \rgta \C_1$ as
\[\bar{f}^j(z_j) = f^j(G_0(z_j)), \ \bar{f}(z) = \prod_{i=1}^t\bar{f}^j(z_j) \]
respectively. Observe that $\bar{f}$ is a \fshape, and
\[ f(G_1(h,z_1,\ldots,z_t)) = \prod_{j=1}^tf^j(G_0(z_j)) = \bar{f}(z).\]

By assumption, we have an explicit generator
$\calG':\zo^{r'} \to [m']^t$ which $\delta'$-fools $(m',t)$-\fshapes.
We claim that $\calG:\hh \times \zo^{r'} \to [m]^n$ defined as
\[ \calG_s(h,w) = G_1(h, \calG'(w))\]
$(\delta' + \delta)$ fools small-variance $(m,n)$-\fshapes.

Since $\calG'$ fools $(m',t)$-\fshapes,
\[ \abs{\E[f(\calG(h,w))] - \E[f(G_1(h,z_1,\ldots,z_t))]} \leq \delta'.\]
By Equation \eqref{eq:g1-works}, whenever $\tvar(f) \leq \log(n/\delta)^C$,
\[\abs{\E[f(G_1(h,z_1,\ldots,z_t))] - \E[f(Z')]} \leq \delta.\]
Combining these equations,
\[\abs{\E[f(\calG(h,w))] - \E[f(Z')]} \leq \delta' + \delta.\]
The seed-length required for $\calG_s$ is  $O(\log(n/\delta))$ for $h$ and $r'$ for $w$.
\end{proof}

%% file: together.tex
% !TEX root = main.tex
\section{Putting things together}

We put the pieces together and prove our main
theorem, Theorem \ref{thm:main}.  We show the following lemma which
allows simultaneous reduction in both the alphabet and the dimension,
going from fooling $(m,n)$-\fshapes\, to  fooling
$(n^2,\lceil\sqrt{n}\rceil)$-\fshapes.

\begin{lem}\label{lem:dimredux}
Let $\delta > 0$, $n > \log^C(1/\delta)$ for some sufficiently large constant $C$, and $t = \lceil \sqrt{n} \rceil$.
If there exists an explicit $\prg$ $\calG'': \zo^{r''} \to [m'']^t$ with
seed-length $r'' = r''(n,\delta)$  which $\delta$-fools
$(m'',t)$-\fshapes\, for all $m'' \leq n^2$, then there exists an explicit generator
$\calG:\zo^r \to [m]^n$ with seed-length  $r = r'' + O(\log(mn/\delta)
\log\log(mn))$ which $4\delta$-fools $(m,n)$-\fshapes.\footnote{Comparing this to Theorem \ref{th:dimreduxlowvar}, the main difference
is that we do not assume that $\tvar(f)$ is small. Further, the
generator $\calG''$ for small dimensions requires $m'' \leq n^2$, and
our goal is to fool \fshapes\ in $n$ dimensions with arbitrary
alphabet size $m$.}
\end{lem}

\begin{proof}
Let $r''$ be the seed-length required for $\calG''$ to have error
$\delta$.   Let  $m' \leq (n/\delta)^c$ be as in the statement of \ref{th:dimreduxlowvar}.
Applying Theorem \ref{th:areduxmain} to $\calG''$, we get a generator
$\calG'$ with seedlength $r'' + O(\log(n/\delta)\log\log(n/\delta))$
that $\delta' = 2\delta$-fools $(m',\sqrt{n})$ \fshapes. Invoking Theorem
\ref{th:dimreduxlowvar} with $\calG'$, we get an explicit
generator $\calG_s:\zo^{r_s} \to [m]^n$ which $3\delta$ fools $(m,n)$-\fshapes\,
$f:[m]^n \to \C_1$ with $\tvar(f) \leq n^{1/9}$ and $m \leq
n^4$, with seed-length
\[ r_s = r'' + O(\log(n/\delta)\log\log(n/\delta)).\]

For $m \leq n^4$, let $\calG_\ell:\zo^{r_\ell} \to [m]^n$ be a
generator for large \fshapes\ as in Theorem
\ref{th:lnorm}, which $\delta$-fools $(m,n)$-\fshapes\, $f:[m]^n \to
\C_1$ with $\tvar(f) \geq C \log^5(1/\delta)$. Since $m \leq n^4$, this generator requires seed-length
\[ r_\ell = O(\log(n/\delta)\log\log(1/\delta)).\]
Define the generator
\[ \calG_{\ell \oplus s}(w_1,w_2) = \calG_\ell(w_\ell) \oplus \calG_s(w_s)\]
where the seeds $w_\ell \in \zo^{r_\ell}$ and $w_s \in \zo^{r_s}$ are
chosen independently and $\oplus$ is interpreted as the sum mod
$m$. Note that the total seed-length is
\[ r_\ell + r_s =  r'' +  O(\log(n/\delta)\log\log(n/\delta)). \]
We now analyze $\calG_{\ell \oplus s}$.
Let $Y = \calG_\ell(w_1)$ and $Z = \calG_\ell(w_2)$ and let $X \in_u [m]^n$. Fix an
$(m,n)$-\fshape\ $f:[m]^n \to \C_1$. We consider two cases based on $\tvar(f)$:

\paragraph{Case 1: $\tvar(f) \geq C \log(1/\delta)^5$.}
For any $z \in [m]^n$, define a new \fshape\,$f_z(y) = f(y \oplus
z)$. Then, for any fixed $z$, $Y$ $\delta$-fools $f_z$ as $\tvar(f_z)
= \tvar(f) \geq C \log(1/\delta)^5$. Therefore,
\[ \abs{\E[f(Y\oplus Z)] - \E[f(X)]} \leq \E_Z \abs{\E[f_Z(Y)] - \E[f(X)]} \leq \delta.\]

\paragraph{Case 2: $\tvar(f) \leq n^{1/9}$.} Consider a
fixing $y$ of $Y$ and define $f_y(Z) = f(y \oplus Z)$. Then, for any
fixed $y$, $Z$ $3\delta $-fools $f_y$ as $\tvar(f_y) \leq n^{1/9}$. Therefore,
\[ \abs{\E[f(Y\oplus Z)] - \E[f(X)]} \leq \E_Y \abs{\E[f_Y(Z)] - \E[f(X)]} \leq 3\delta.\]

In either case, we have
\[ \abs{\E[f(Y \oplus Z)] - \E[f(X)]} \leq 3\delta.\]

Finally, for arbitrary $m$, by applying Theorem \ref{th:areduxmain} to $\calG_{\ell
  \oplus s}$, we get a generator $\calG:\zo^r \rgta [m]^n$  that
$4\delta$ fools $(m,n)$-\fshapes\  with seed-length
\[ O(\log (m/\delta)\log \log(m)) + r_\ell +r_s =  r'' +  O(\log(nm/\delta)\log\log(nm/\delta)). \]
\end{proof}

We prove Theorem \ref{thm:main} by repeated applications of this lemma.

\begin{proof}[Proof of Theorem \ref{thm:main}]
Assume that the final error desired is $\delta'$. Let $\delta =
\delta'/4\log\log(n)$.  Applying Lemma \ref{lem:dimredux}, by
using $O(\log(mn/\delta')\log \log (mn/\delta'))$ random bits
we reduce fooling $(m,n)$-\fshapes \ to fooling
$(m',\lceil\sqrt{n}\rceil)$-\fshapes\, for $m' \leq n^2$.

We now apply the lemma $O(\log \log n)$ times to reduce to the
case of fooling $(\log^C(1/\delta),\log^C(1/\delta))$-\fshapes.
This can be done by noting that by Lemma \ref{lem:bounded} it suffices to fool \fshapes\, with $\log(f_i)$ having $O(\log(1/\delta))$ bits of precision. Such \fshapes\, can be computed by width-$O(\log(1/\delta))$ ROBPs, and thus using the generator from Theorem \ref{th:inwprg}, we can fool this case with seed length $O(\log(1/\delta)\log\log(1/\delta))$ bits.
Since each step requires $O(\log(n/\delta)\log\log(n/\delta)$
random bits, the overall seedlength is bounded by
\[ O(\log(mn/\delta)\log \log (mn/\delta) + O(\log(n/\delta)(\log\log(n/\delta))^2).\]
\end{proof}

%% file: applications.tex
% !TEX root = main.tex
\section{Applications of $\prg$s for Fourier shapes}
\label{sec:applications}

In this Section, we show how Theorem \ref{thm:main} implies near
optimal $\prg$s for halfspaces, modular tests and combinatorial
shapes. We first prove two technical lemmas relating closeness between
Fourier transforms of integer valued random variables to closeness
under other metrics. We define the Fourier distance, statistical distance and Kolmogorov
distance between two integer-valued random variables respectively as
\begin{align}
d_{FT}(Z_1,Z_2) & = \max_{\alpha \in [0,1]} \left|\E[\exp(2\pi i
  \alpha Z_1)] - \E[\exp(2\pi i \alpha Z_2)]\right|, \label{eq:f-dist} \\
d_{TV}(Z_1,Z_2) & = \frac{1}{2}\sum_{j \in \Z}|\pr(Z_1= j) - \pr(Z_2 = j)|,\label{eq:tv-dist}\\
d_K(Z_1,Z_2) & = \max_{k \in \mathbb{Z}}(|\pr(Z_1\leq k) -
 \pr(Z_2\leq k)|) \label{eq:cdf-dist}
\end{align}

The first standard claim relates closeness in statistical distance
and Fourier distance for bounded integer valued random variables.

\begin{lem}\label{lem:fouriertotv}
Let $Z_1, Z_2$ be two integer-valued random variables supported on $[0,N]$. Then,
\[ d_{TV}(Z_1,Z_2) \leq O(\sqrt{N})\cdot d_{FT}(Z_1,Z_2) .\]
\end{lem}
\begin{proof}
Note that the distribution $Z_1-Z_2$ is supported on at most $4N+1$ points. Therefore,
\[
\dtv(Z_1,Z_2) = \|Z_1-Z_2\|_1 \leq \sqrt{4N+1}\|Z_1-Z_2\|_2.
\]
On the other hand, the Plancherel identity implies that
\[
\|Z_1-Z_2\|_2 \leq d_{FT}(Z_1,Z_2).
\]
This completes the proof.
\end{proof}

The second claim relates closeness in Kolmogorov distance to
closeness in Fourier distance. The key is that unlike in Lemma
\ref{lem:fouriertotv}, the dependence on $N$ is logarithmic. This difference is crucial to
fooling halfspaces with polynomially small error (since there $N$ can
exponential in the dimension  $n$).

%The starting point of our analysis and construction is to note that showing closeness in Kolmogorov for discrete random variables of relatively small support is equivalent to showing that the \emph{Fourier transforms} of the random variables are close. This will allow us to use various analytic tools. Concretely, we shall use the following elementary fact about the discrete Fourier transform.

\begin{lem}\label{lem:introfouriertocdf}
Let $Z_1, Z_2$ be two integer-valued random variables supported on $[-N,N]$. Then,
\[ d_K(Z_1,Z_2) \leq O(\log(N) \cdot d_{FT}(Z_1,Z_2)) .\]
\end{lem}
\begin{proof}
By definition we have that
$$
d_K(Z_1,Z_2) = \max_{-N\leq k \leq N}(|\pr(Z_1\leq k) - \pr(Z_2\leq k)|).
$$
We note that
\begin{align*}
\pr(Z_i\leq k) & = \sum_{j=-N}^k \pr(Z_i=j)\\
& = \sum_{j=-N}^k \int_0^1 \exp(-2\pi i j \alpha) \E[\exp(2\pi i \alpha Z_i)] d\alpha\\
& = \int_0^1 s(k,N,\alpha)\E[\exp(2\pi i \alpha Z_i)] d\alpha
\end{align*}
where
\[ s(k,N,\alpha) = \sum_{j=-N}^k  \exp(-2\pi i j \alpha) .\]
It is clear that $|s(k,N,\alpha)| \leq 2N$. Further,
\begin{align*} |s(k,N,\alpha)| & = \left|\frac{\exp(-2\pi i k\alpha)
  (\exp(2\pi i (N+k+1)\alpha)-1)}{\exp(2\pi i \alpha)-1}\right|  \leq \frac{1}{|\exp(2\pi i \alpha)
    -1|} \leq O(\frac{1}{[\alpha]})
\end{align*}
where $[\alpha]$ is the distance between $\alpha$ and the nearest integer.
Therefore, we have
\begin{align*}
|\pr(Z_1\leq k) - \pr(Z_2\leq k)|  & \leq \int_0^1 |s(k,N,\alpha)|\left|\E[\exp(2\pi i \alpha Z_1)]-\E[\exp(2\pi i \alpha Z_1)] \right|d\alpha \\
& \leq \int_0^1 O\left(\min\left(N,\frac{1}{[\alpha]}\right)\right)d_{FT}(Z_1,Z_2) d\alpha\\
& = O(d_{FT}(Z_1,Z_2)) \left(\int_0^{1/N} Nd\alpha + \int_{1/N}^{1/2}\frac{d\alpha}{\alpha} + \int_{1/2}^{1-1/N} \frac{d\alpha}{1-\alpha} + \int_{1-1/N}^1 N d\alpha\right)\\
& = O(d_{FT}(Z_1,Z_2) \log(N)).
\end{align*}
\end{proof}

\subsection{Corollaries of the main result}

We combine Lemma \ref{lem:introfouriertocdf} with Theorem
\ref{thm:main} to derive Corollary \ref{cor:halfspaces}, which gives
$\prg$s for halfspaces with polynomially small error from $\prg$s for
$(2,n)$-\fshapes. 

\begin{proof}[Proof of Corollary \ref{cor:halfspaces}]
Let $\calG:\zo^r \to \dpm^n$ be a $\prg$ which $\delta$-fools
$(2,n)$-\fshapes\,(here we identify $[2]$ with $\dpm$ arbitrarily). We
claim that $\calG$ also fools all halfspaces with error at most $\eps
= O(n \log(n)\delta)$.

Let $h:\dpm^n \to \dpm$ be a halfspace given by $h(x) =
\sgn(\iprod{w}{x} - \theta)$. It is well known that we can assume the
weights and the threshold $\theta$ to be integers bounded in the range
$[-N,N]$ for $N = 2^{O(n \log n)}$ (cf.~\cite{LewisC67}). Let $X \in_u \dpm^n$ and $Y =
\calG(y)$ for $y \in_u \zo^r$ and $Z_1 = \iprod{w}{X}$, $Z_2 =
\iprod{w}{Y}$. Note that $Z_1, Z_2$ are bounded in the range $[-n
  \cdot N, n \cdot N]$.

We first claim that
\[ d_{FT}(Z_1,Z_2) \leq \delta.\]
For $\alpha \in [0,1]$, we define $f_\alpha:\dpm^n \to \C_1$ as
\begin{align}
\label{eq:lin-fshape}
f_\alpha(x) = \exp(2 \pi i \alpha \iprod{w}{x}) = \prod_{j=1}^n
\exp(2\pi i \alpha w_j x_j)
\end{align}
then $f_\alpha$ is a $(2,n)$-\fshape. Hence,
\[ |\E[f_\alpha(X)] - \E[f_\alpha(Y)]| \leq \delta.\]
That $d_{FT}(Z_1,Z_2) \leq \delta$ now follows from the definition of Fourier distance,
and the fact that $\E[f_\alpha(X)]$ and $\E[f_\alpha(Y)]$ are the Fourier
transforms of $X$ and $Y$ at $\alpha$ respectively.

Therefore, by Lemma \ref{lem:introfouriertocdf} applied to $Z_1$, $Z_2$, $d_K(Z_1, Z_2) \leq O(n \log n) \delta$. Finally, note that
$$\abs{\E[h(X)] - \E[h(Y)]} \leq d_K(\iprod{w}{X}, \iprod{w}{Y}) = d_K(Z_1,Z_2) \leq O(n\log n) \delta.$$

The corollary now follows by picking a generator as in Theorem
\ref{thm:main} for $m=2$ with error $\delta = \epsilon/(C n \log n)$
for sufficiently big $C$.
\end{proof}

To prove Corollary \ref{cor:halfspaces2}, we need the following lemma about generalized
halfspaces.

\begin{lem}
\label{lem:gen-halfspace}
In Definition \ref{def:gh}, we may assume that each $g_i(j)$ is an
integer of absolute value $(mn)^{O(mn)}$.
\end{lem}
\begin{proof}
Let $g:[m]^n \to \zo$ be a generalized halfspace where the $g_i$s are
arbitrary. Embed $[m]^n$ into $\zo^{mn}$ by sending each $x_i \in [m]$ to
$(y_{i,1},\ldots,y_{i,m})$ where $y_{i,j} = 1$ if $x_i = j$ and
$y_{i,j} =0$ otherwise. Note that 
\[
\sum_{i=1}^n g_i(x_i) = \sum_{i =1}^n\sum_{j = 1}^m g_i(j)y_{i,j}
\]
However, the halfspace
\[ \sum_{i =1}^n\sum_{j = 1}^m g_i(j)y_{i,j} \geq \theta \]
over the domain $\zo^{mn}$ has a representation where the weights
$g'_i(j)$ and $\theta'$ are integers of size at most
$(mn)^{O(mn)}$. Hence we can replace each $g_i(j)$ in the defintion of
$g$ with $g_i'(j)$ without changing its value at any point in $[m]^n$.
\end{proof}

We now prove Corollary \ref{cor:halfspaces2} giving $\prg$s for
generalized halfspaces over $[m]^n$.
\begin{proof}[Proof of Corollary \ref{cor:halfspaces2}]
Letting $X\in_u [m]^n$ and letting $X'$ be obtained from a PRG for
$(m,n)$-\fshape s with error at most $\epsilon$, we let $Z_1=\sum_i
g_i(X_i)$ and $Z_2=\sum_i g_i(X_i')$. By Lemma
\ref{lem:introfouriertocdf} that $d_K(Z_1,Z_2)\leq O(\epsilon nm
\log(nm))$. Picking $\epsilon$ sufficiently small gives our generator
for generalized halfspaces. 
\end{proof}

Next we use Corollary \ref{cor:halfspaces2}
to get $\prg$s fooling halfspaces under general product distributions. From the definition of
generalized halfspaces, it follows that if $\calD$ is a discrete product distribution on
$\R^n$ where each co-ordinate can be sampled using $\log(m)$ bits,
then fooling halfspaces under $\calD$ reduces to fooling generalized
halfspaces over $[m]^n$ for some suitable choice of $g_i$. In fact
\cite{GopalanOWZ10} showed that fooling such distributions is in fact
sufficient to sufficient to fool continuous product distributions with
bounded moments. The following is a restatement of \cite[Lemma
  6.1]{GopalanOWZ10}. 

\begin{lem}
\label{lem:gowz}
Let $X$ be a product distribution on $\R^n$ such that for all $i \in [n]$,
\[ \E[X_i] =0, \E[X_i^2] =1, \E[X_i^4] \leq C. \]
Then there exists a discrete product distribution $Y$ such that for
every halfspace $h$,
\[ \abs{\E[h(X)] - \E[h(Y)]} \leq \eps.\]
Further, each $Y_i$ can be sampled using $\log(n,1/\eps,C)$ random bits.
\end{lem}

Note that the first and second moment conditions on $X$ can be
obtained for any product distribution by an affine transformation. Hence
we get Corollary \ref{cor:general} from combining Lemma \ref{lem:gowz} with
Corollary \ref{cor:halfspaces2}. In particular, there exist generators that fool all halfspaces with
error $\eps$ under the Gaussian distribution with seed-length $r =
O(\log(n/\eps)(\log\log(n/\eps))^2)$. This nearly matches the recent
result of \cite{KothariM14} upto a $\log\log$ factor. Further, it is
known (see e.g \cite[Lemma 11.1]{GopalanOWZ10}) that $\prg$s for
halfspaces under the Gaussian distribution imply $\prg$s for
halfspaces over the sphere.

We next prove Corollary \ref{cor:ch-bound} which derandomizes
the Chernoff bound.

\begin{proof}[Proof of Corollary \ref{cor:ch-bound}]
First note that we can assume without loss of generality that each $X_i$ can be sampled with $r_x = O(\log(mn/\epsilon))$ bits (by ignoring elements which happen with smaller probability). In particular, let each $X_i$ have the same distribution as $h_i(Z)$ for $Z \in_u [m']$ where $m' = 2^{r_x}$ (here we identify $[m']$ with $\zo^{r_x}$) and some function $h_i:[m'] \to [m]$. Let $\calG:\zo^r \to [m']^n$ be a PRG which $(\epsilon/2)$-fools $(m',n)$-generalized halfspaces. Now, let $Y = (h_1(Z_1),h_2(Z_2),\ldots,h_n(Z_n))$, where $(Z_1,\ldots,Z_n) = \calG(w)$ for $w \in_u \zo^r$. 

Note that $Y$ can be sampled with $O(\log(mn/\epsilon) \cdot (\log\log^2(mn/\epsilon)))$ random bits. We claim that $Y$ satisfies the required guarantees. To see this, define the generalized halfspaces 
\[g_+(z) = \sgn\bkets{\sum_{i=1}^ng_i(h_i(z_i)) - \theta},\;\;\; g_-(z) = \sgn\bkets{\sum_{i=1}^n -g_i(h_i(z_i)) + \theta} ,\]
where 
\[ \theta = t + \sum_i\E[g_i(X_i)] = t +  \sum_{i=1}^n\E[g_i(Y_i)].\] 
From Corollary \ref{cor:halfspaces2} it follows that 
\[ |\E[g_+(X)] - \E[g_+(Y)]| \leq \eps/2, \;\;\; |\E[g_-(X)] - \E[g_-(Y)]| \leq \eps/2. \]
From the Chernoff-Hoeffding bound \cite{Hoeffding}, we have
\[ \E[g_+(X) + g_-(X)] = \Pr\sbkets{\abs{\sum_{i=1}^ng_i(X_i) - \sum_{i=1}^n\E[g_i(X_i)]}
    \geq t} \leq 2e^{-t^2/2n}. \]
Hence by the triangle inequality,
\[ \Pr\sbkets{\abs{\sum_{i=1}^ng_i(Y_i) - \sum_{i=1}^n\E[g_i(Y_i)]}
    \geq t}  = \E[g_+(Y) + g_-(Y)] \leq 2e^{-t^2/2n} + \eps. \]
\end{proof}

We next prove Corollary \ref{cor:modular} about fooling modular tests.
\begin{proof}[Proof of Corollary \ref{cor:modular}]
Let $\calG:\zo^r \to \zo^n$ be a $\prg$ which fools $(2,n)$-\fshapes\ with
error $\epsilon/\sqrt{Mn}$. We claim that $\calG$ fools modular tests
with error at most $\epsilon$.

Let $g(x) = \mathds{1}(\sum_i a_i x_i \mod M \in S)$ be a modular
test, let $X \in_u \zo^n$ and $Y = \calG(y)$ for $y \in_u
\zo^r$. In order to fools modular tests, it suffices that
\[ d_{TV}(\sum_i a_i X_i, \sum_i a_i Y_i) \leq \epsilon.\]
On the other hand, since both these random variables are bounded in
the range $\{0,Mn\}$, by Lemma \ref{lem:fouriertotv}
\[ d_{TV}\bkets{\sum_i a_i X_i, \sum_i a_i Y_i} \leq \sqrt{Mn}\cdot
d_{FT}\bkets{\sum_i a_i X_i, \sum_i a_i Y_i} \leq \eps\]
where the last inequality uses the fact that the Fourier transforms of
both random variables are $(2,n)$-\fshapes\ by Equation \eqref{eq:lin-fshape}.
\end{proof}

Next we prove Corollary \ref{cor:cshapes} giving $\prg$s from
combinatorial shapes.
\begin{proof}[Proof of Corollary \ref{cor:cshapes}]
Recall that a combinatorial shape $f:[m]^n \to \zo$ is a function
\[ f(x) = h\bkets{\sum_{i=1}^ng_i(x_i)} \]
where $g_i:[m] \to \zo$ and $h:\{0,\ldots, n\} \to \zo$.  Since
$\sum_ig_i(x_i) \in \{0,\ldots,n\}$, it suffices to fool the generalized
halfspaces
\[ f(x) = \sum_ig_i(x_i) -\theta \]
for $\theta \in \{0,\ldots,n\}$ each with error $\eps/n$.
Hence the claim follows from Corollary \ref{cor:halfspaces2} about
fooling generalized halfspaces.
\end{proof}

\ignore{
\begin{proof}[Proof of Corollary \ref{cor:gaussian}]
Essentially the only difficulty here is to reduce to the case of a halfspace on a discrete rather than continuous domain. For any $\delta>0$, this it is easy to construct a function $f:[m]\rightarrow \R$ with $m=O(\log(1/\delta))$ so that for $X\in_u [m]$ then there is a correlated Gaussian $Y$ so that $|f(X)-Y|<\delta$ with probability at least $1-\delta$ (see \cite{K15}, Lemma 5). Letting $\delta = \epsilon/n^2$ and letting $X_i\in_u [m]$ and $Y_i$ correlated Gaussians, we note that with probability better than $1-\epsilon/2$, that $|f(X_i)-Y_i|<\delta$ for all $i$. If $v$ is any unit vector, this would imply that $\left| \sum v_i f(X_i) - \sum v_i Y_i\right| \ll \epsilon$. Together these imply that the different in probabilities that $\sum v_i Y_i \leq \theta$ and $\sum v_i f(X_i) \leq \theta$ differ by $O(\epsilon)$. Therefore, it suffices to find a random variable $X'$ so that
$$
\left|\pr\left( \sum v_i f_i(X_i) < \epsilon \right) - \pr\left( \sum v_i f_i(X_i') < \epsilon \right)\right| < \epsilon.
$$
By Corollary \ref{cor:halfspaces2}, this can be done with seedlength $O(\log(n/\epsilon)\log\log^2(n/\epsilon))$. This completes the proof.
\end{proof}
}

%% file: appendix.tex
\section{Proofs from Section \ref{sec:prelims}}

\label{sec:appendix}
\begin{comment}
\begin{proof}[Proof of Lemma \ref{hcBiasedLem}]
This follows from the fact that $\|Q^p\|_1 \leq \|Q\|_1^p.$ Therefore,
\[
\E\sbkets{ Q(x)^p} \leq \E_{X\in_u \{\pm 1\}^n}[Q(X)^p] +\|Q\|_1^p
\delta\leq (p-1)^{p d/2} \cdot \|Q\|_2^p + \|Q\|_1^p \delta.
\]
\end{proof}

\begin{proof}[Proof of Lemma \ref{lm:epschernoff}]
Note that because $\|v\|_2=1$ that $\|v\|_1\leq \|v\|_0^{1/2}$ by Cauchy-Schwarz. We note by the Markov inequality that for even $p$ that
$$
\pr[|\iprod{v}{x}| > t] \leq t^{-p}\E[|\iprod{v}{x}|^p].
$$
We need a slightly strengthened version of Lemma \ref{hcBiasedLem} to bound this.
Note that if $f(x)=\iprod{v}{x}$
$$
\E[f(x)] \leq \|f^p\|_0\epsilon + \|f\|_p^p \leq \|v\|_0^p \epsilon + (p-1)(p-3)\cdots 1.
$$
The bound on $\|f\|_p$ comes from noting that the expectation of $f^p$ under Gaussian inputs is $(p-1)(p-3)\cdots 1$ and that the expectation under Bernoulli inputs is at most this (which can be seen by expanding and comparing terms). Therefore, we have that
$$
\pr[|\iprod{v}{x}| > t] \leq t^{-p}\sqrt{2}(p/e)^{p/2} + t^{-p}\|v\|_0^p \epsilon
$$
Letting $p$ be the largest even integer less than $t^2$, we find that this is at most
$$
\sqrt{2} \exp(-p/2)+ \|v\|_0^{t^2} \epsilon,
$$
which is sufficient when $t\geq 2$. For $1\leq t \leq \sqrt{2}$, the
trivial upper bound of $1$ is sufficient, and for $\sqrt{2}\leq t \leq
2$, we may instead use the bound for $p=2$.
\end{proof}
\end{comment}

\begin{proof}[Proof of Lemma \ref{hashmomentsLem}]
Let $I_{i,k}$ be the indicator function of the event that
$h(i)=k$. Note that $h(v)=\sum_{i,j,k} I_{i,k}I_{j,k}v_i^2 v_j^2.$
Therefore,
$$
h(v)^p = \sum_{i_1,\ldots,i_{p},j_1,\ldots,j_p}\sum_{k_1,\ldots,k_p} \prod_{t=1}^p I_{i_{t},k_t}I_{j_{t},k_t} \prod_{t=1}^{p} v_{i_t}^2v_{j_t}^2.
$$
Let $R(i_t,j_t,k_t)$ be $0$ if for some $t,t'$ $k_t\neq k_t'$ but one
of $i_{t}$ or $j_t$ equals $i_{t'}$ or $j_{t'}$ and otherwise be equal
to $m^{-T}$ where $T$ is the number of distinct values taken by $i_t$
or $j_t$. Notice that by the $\delta$-biasedness of $h$ that
$$
\E\left[\prod_{t=1}^p I_{i_{t},k_t}I_{j_{t},k_t}\right] \leq R(i_t,j_t,k_t) + \delta.
$$
Combining with the above we find that
\begin{align*}
\E[h(v)^p] & \leq \sum_{i_1,\ldots,i_{p},j_1,\ldots,j_p}\sum_{k_1,\ldots,k_p} (R(i_t,j_t,k_t) + \delta)\prod_{t=1}^{p} v_{i_t}^2v_{j_t}^2\\
& \leq \sum_{i_1,\ldots,i_{p},j_1,\ldots,j_p}\sum_{k_1,\ldots,k_p} R(i_t,j_t,k_t)\prod_{t=1}^{p} v_{i_t}^2v_{j_t}^2 + \delta m^p\sum_{i_1,\ldots,i_{p},j_1,\ldots,j_p}\prod_{t=1}^{p} v_{i_t}^2v_{j_t}^2\\
& \leq \sum_{i_1,\ldots,i_{p},j_1,\ldots,j_p}\sum_{k_1,\ldots,k_p} R(i_t,j_t,k_t)\prod_{t=1}^{p} v_{i_t}^2v_{j_t}^2 + \delta m^p \|v\|_2^{4p}.
\end{align*}
Next we consider
$$
\sum_{k_1,\ldots,k_p} R(i_t,j_t,k_t)
$$
for fixed values of $i_1,\ldots,i_p,j_1,\ldots,j_p$. We claim that it
is at most $m^{-S/2}$ where $S$ is again the number of distinct
elements of the form $i_t$ or $j_t$ that appear in this way an odd
number of times. Letting $T$ be the number of distinct elements of the
form $i_t$ or $j_t$, the expression in question is $m^{-T}$ times the
number of choices of $k_t$ so that each value of $i_t$ or $j_t$
appears with only one value of $k_t$. In other words this is $m^{-T}$
times the number of functions $f:\{i_t,j_t\}\rightarrow[m]$ so that
$f(i_t)=f(j_t)$ for all $t$. This last relation splits $\{i_t,j_t\}$
into equivalence classes given by the transitive closure of the
operation that $x\sim y$ if $x=i_t$ and $y=j_t$ for some $t$. We note
that any $x$ that appears an odd number of times as an $i_t$ or $j_t$
must be in an equivalence class of size at least $2$ because it must
appear at least once with some other element. Therefore, the number of
equivalence classes, $E$ is at least $T-S/2$. Thus, the sum in
question is at most $m^{-T}m^E \leq m^{-S/2}$. Therefore, we have that
$$
\E[h(v)^p] \leq (2p)!\sum_{\textrm{Multisets }M\subset [n],
  |M|=2p}m^{-\{\mathrm{Odd}(M)\}/2}\prod_{i\in M} v_i^2 +\delta m^p
\|v\|_2^{4p}.
$$
Where $\mathrm{Odd}(M)$ is the number of elements occurring in $M$ an
odd number of times. This equals
\begin{align*}
\E[h(v)^p] & \leq (2p)!\sum_{k=0}^p\sum_{\textrm{Multisets }M\subset
  [n], |M|=2p, \mathrm{Odd}(M)=2k}m^{-k}\prod_{i\in M} v_i^2 +\delta
m^p \|v\|_2^{4p}\\
& \leq (2p)!\sum_{k=0}^p m^{-k}
\sum_{i_1,\ldots,i_{2k}}\sum_{j_1,\ldots,j_{p-k}}\prod v_{i_t}^2 \prod
v_{j_t}^4+\delta m^p \|v\|_2^{4p}\\
& = (2p)!\sum_{k=0}^p \left(\frac{\|v\|_2^4}{m}\right)^k
\|v\|_4^{4(p-k)}+\delta m^p \|v\|_2^{4p}\\
& \leq O(p)^{2p} \left(\frac{\|v\|_2^4}{m}\right)^p + O(p)^{2p}
\|v\|_4^{4p}+\delta m^p \|v\|_2^{4p}.
\end{align*}
Note that the second line above comes from taking $M$ to be the
multiset
$$\{i_1,i_2,\ldots,i_{2k},j_1,j_1,j_2,j_2,\ldots,j_{p-k},j_{p-k}\}.$$
This completes our proof.
\end{proof}

\begin{proof}[Proof of Lemma \ref{lm:hashing}]
 Let $X_i$ denote the indicator random variable which is $1$ if $h(i) = j$ and $0$ otherwise. Let $Z = \sum_i v_i X_i$. Now, if $h$ were a truly random hash function, then, by Hoeffding's inequality,
$$\pr\sbkets{| Z - \nmo{v}/m| \geq t} \leq 2 \exp\bkets{- t^2/2\sum_i v_i^2}.$$
Therefore, for a truly random hash function and even integer $p \geq 2$, $\nmp{Z} = O(\nmt{v}) \sqrt{p}$. Therefore, for a $\delta$-biased hash family, we get $\nmp{Z}^p \leq O( p)^{p/2} \nmt{v}^p + \nmo{v}^p \delta$. Hence, by Markov's inequality, for any $t > 0$,
$$\pr\sbkets{ |Z - \nmo{v}/m| \geq t} \leq \frac{O(p)^{p/2} \nmt{v}^p + \nmo{v}^p \delta}{t^p}.$$
\end{proof}

\section{Proofs from Section \ref{sec:maintech}}
\label{app:maintech}

\begin{proof}[Proof of Lemma \ref{lm:mombound}]
First we note that since for any complex random variable, $Z$, that
$$
\E\left[|Z|^k \right] = 2^{O(k)}\E[|\Re(Z)|^k+|\Im(Z)|^k]
$$
and $\var(Z) = \var(\Re(Z))+\var(\Im(Z))$, it suffices to prove our lemma when $Z$ is a real-valued random variable.

We can now compute the expectation of $\left(\sum_i Z_i\right)^k$ by expanding out the polynomial in question and computing the expectation of each term individually. In particular, we have that
\begin{align*}
\E\sbkets{\abs{\sum_i Z_i}^k} & = \sum_{i_1,\ldots,i_k} \E\left[\prod_{j=1}^k Z_{i_j} \right].
\end{align*}
Next we group the terms above by the set $S$ of indices that occur as $i_j$ for some $j$. Thus, we get
$$
\sum_{m=1}^k \sum_{|S|=m} \sum_{\substack{i_1,\ldots,i_k\in S\\ \{i_j\}=S}} \E\left[\prod_{j=1}^k Z_{i_j} \right].
$$
We note that the expectation in question is 0 unless for each $j\in S$, $Z_j$ occurs at least twice in the product. Therefore, the expectation is 0 unless $m\leq k/2$ and overall is at most $B^{k-2m}\prod_{j\in S} \var(Z_j)$. Thus, the expectation in question is at most
$$
\sum_{m=1}^{k/2} \sum_{|S|=m} m^k B^{k-2m}\prod_{j\in S} \var(Z_j).
$$
Next, note that by expanding out $\left(\sum_i \var(Z_i) \right)^m$ we find that $\sigma^{2m} \geq m!\sum_{|S|=m} \prod_{j\in S} \var(Z_j).$ Therefore, the expectation in question is at most
\begin{align*}
\sum_{m=1}^{k/2} 2^{O(k)} m^{k-m} B^{k-2m}\sigma^{2m} & \leq 2^{O(k)} \sum_{m=0}^{k/2} k^{k-m}B^{k-2m}\sigma^{2m}\\
& \leq 2^{O(k)}\left(k^{k/2}\sigma^k + k^kB^k\right)\\
& \leq 2^{O(k)}(\sigma \sqrt{k} + Bk)^k,
\end{align*}
as desired.
\end{proof}

\section{Proofs from Section \ref{sec:largenorm}}
\label{app:largenorm}

\begin{proof}[Proof of Lemma \ref{lem:pw-sampling}]
First note that $t \in \{0,\ldots,T\}$ satisfying the hypothesis
exists since $\nmt{v}^2 \in [1,n]$.
For $\ell \in [n]$, let $I(\ell)$ be the indicator random variable which is $1$ if $\ell
\in B_t$. Since $|B_t| =2^t$, $\Pr[I(\ell) =1] =2^t/n$. If we set $V = \nm{v_t}^2$,
\begin{align*}
V & = \sum_\ell v_\ell^2 I(\ell)\\
\E[V] & = \nm{v}^2 \frac{2^t}{n} \in [1/2,1].
\end{align*}
By the pairwise independence of $\sigma$,
\begin{align*}
\E[V^2] & = \sum_{\ell=1}^n v_\ell^4 I(\ell) + \sum_{\ell \neq
  \ell'=1}^n v_\ell^2 v_{\ell'}^2 I(\ell)I(\ell')\\
& \leq \sum_{\ell=1}^n v_\ell^4 \frac{2^t}{n} + \sum_{\ell \neq
  \ell'=1}^n v_\ell^2 v_{\ell'}^2 \frac{2^{2t}}{n}\\
&  \leq \frac{2^t}{n} \nm{v}_4^4 + \E[V]^2.
\end{align*}
Therefore,
\[ \var(V) = \E[V^2] - \E[V]^2  \leq \frac{2^t}{n}\nm{v}^4 \leq \frac{2^t\nmt{v}^2}{n} \nm{v}_\infty^2 \leq \frac{1}{16}\]
Thus, by Chebyshev's inequality,
\[\pr[|V - \E[V]| > 1/3] \leq 9/16\]
In particular, with probability at least $7/16$, $V = \nmt{v^t}^2 \in[1/6,4/3]$.
\end{proof}

\begin{proof}[Proof of Lemma \ref{SpreadingHashLem}]
Let $\ell = 2 C_1 \log(1/\delta)$ and $T=\Theta(\log^5(1/\delta))$ to
be chosen later. Let $\hh = \{h:[n]\rightarrow[T]\}$ to be a
$\delta'$-biased family for $\delta'=\exp(-C(\log(1/\delta)))$ for $C$ a
sufficiently large constant.

Let $p=c\log(1/\delta)/\log\log(1/\delta))$ for a constant $c$ to be
chosen later. Let $v \in [0,1]^n$ with $\nmt{v}^2 \geq C_2
\log^5(1/\delta)$ and note that if $\nmt{v_{h^{-1}(j)}}^2 \geq
\nmt{v}^2/\ell$ for some $j \in [T]$, then $h(v) \geq
\nmt{v}^4/\ell^2$ (recall the definition of $h(v)$ from Equation
\eqref{eq:hv}). Therefore, by Lemma \ref{hashmomentsLem} and Markov's
inequality, the probability that this happens is at most
\begin{align*}
\frac{\E[h(v)^p]\ell^{2p}}{\nmt{v}^{4p}} &\leq \bkets{\frac{\ell^{2p}}{\nmt{v}^{4p}}}\bkets{O(p)^{2p} \left(\frac{\nmt{v}^4}{T}\right)^p + O(p)^{2p} \|v\|_4^{4p}+ T^p \nmt{v}^{4p}\delta'}\\
&\leq O\bkets{\frac{p^2 \ell^2}{T}}^p + O\bkets{\frac{p^2 \ell^2}{\nmt{v}^2}}^p + T^p \ell^{2p} \delta'\\
&\leq O(\log(1/\delta))^{-p} + O(\log(1/\delta))^{7p} \delta'\\
&< \delta,
\end{align*}
for a suitable choice of the constant $c$ and $\delta' = \exp(-C \log(1/\delta))$.

Now suppose that $\nmt{v_{h^{-1}(j)}}^2 < \nmt{v}^2/\ell$ for all $j
\in [T]$. Let $I = \{j: \nm{v_{h^{-1}(j)}}^2 \geq
\nmt{v}^2/2T\}$. Then,
\[ \nmt{v}^2 \leq |I| \cdot (\nmt{v}^2/\ell) + T \cdot \nmt{v}^2/(2T).\]
Therefore, we must have $|I| \geq \ell/2$. This proves the claim.
\end{proof}

\section{Proofs from Section \ref{sec:dimredux}}
\label{app:dimredux}
\begin{proof}
Let $\alpha=n^{-1/3},\beta=n^{-1/36}$.

Note that $|L| \leq \tvar(f)/\alpha\leq n^{2/9}$. Since $h \in_u \hh$ is $k$-wise independent, for any index $j \in [t]$,
\begin{align}
\label{eq:bound_1}
\pr[|L \cap h^{-1}(j)| > k/2] \leq \binom{|L|}{k/2}
\left(\frac{1}{t}\right)^{k/2} \leq
\left(\frac{\tvar(f)}{\alpha}\right)^{k/2}
\left(\frac{1}{t}\right)^{k/2}\leq O\left(n^{-5/18} \right)^{k/2}.
\end{align}

Define $v \in \R^n$ by $v_j = \sigma^2(f_j)$ if $j \in S$ and $0$
otherwise. Now,
\[ \nm{v}_2^2 = \sum_{j \in S} \sigma^4(f_j) \leq \max_{j \in S}\sigma^2(f_j)\sum_{j \in S} \sigma^2(f_j) \leq
\tvar(f) \alpha.\]
By Lemma \ref{lm:hashing} applied to $v$, we get that for any $j \in [t]$,
\begin{align}
\pr_{h \in \hh}\left[\sum_{\ell \in S: h(\ell) = j} \sigma^2(f_\ell) \geq
  \frac{\tvar(f)}{t} + \frac{(\tvar(f))^{1/2} \alpha^{1/4}}{2} \leq \beta \right] \leq O(k)^{k/2}
  \alpha^{k/4}=O(k)^{k/2}n^{-\Omega(k)}.
\end{align}
This completes the proof.
\end{proof}